\def\llncs{0}
\def\mnotes{0}
\documentclass[11pt]{article}

\usepackage{amsfonts,amsmath,epsfig, graphicx, fullpage,mathrsfs}
\usepackage{comment}

\usepackage{color,bm}

\def\colorson{0}

\usepackage{times,color}
\usepackage[usenames,dvipsnames,svgnames,table]{xcolor}
\usepackage{amsfonts,amsmath}
\usepackage{url}
\usepackage{fullpage}

 \ifnum\llncs=0

\colorlet{notgreen}{blue!20!black!30!green}

\usepackage{amsthm}
\usepackage{latexsym,amssymb}
\usepackage{graphicx}

   \usepackage[letterpaper=true,pdftex,
    bookmarks,
    plainpages=false, 
    pdfpagelabels=true, 
    colorlinks=true,breaklinks=true,linkcolor=red,citecolor=green,
    filecolor=black,menucolor=black,pagecolor=black,urlcolor=blue
    pdftitle=UW\ LaTeX\ Thesis\ Template, 
    pdfauthor=Jalaj Upadhyay	
    ]{hyperref}{}

\fi 
\ifnum\mnotes=0
\newcommand{\mnote}[1]{}
\else
\newcounter{mynotes}
\newcommand{\mnote}[1]{\addtocounter{mynotes}{1}{{\bf !}}%
\marginpar{\scriptsize  {\arabic{mynotes}.\ {\sf \textcolor{red}{#1}}}}}
\fi

\ifnum\colorson=1
\newenvironment{todo}{\noindent
\sf \footnotesize \textcolor{blue}{To go here}:
\begin{CompactItemize}\color{blue}}
{\color{black}\end{CompactItemize}\rm \normalsize}
\else

\fi

\ifnum\llncs=0

\newenvironment{CompactItemize}{
  \vspace{-5pt}
 \begin{list}{$\bullet$}{%
      \setlength{\leftmargin}{8pt}%
      \setlength{\itemsep}{-2pt}
      }}
{\end{list}}

\makeatletter
\def\thm@space@setup{\thm@preskip=1pt
\thm@postskip=1pt}
\makeatother
\newtheoremstyle{newstyle}      
{} 
{} 
{\mdseries} 
{} 
{\bfseries} 
{.} 
{ } 
{} 

\theoremstyle{newstyle}

\newtheorem{theorem}{Theorem}
\newtheorem{lemma}[theorem]{Lemma}

\newtheorem{corollary}[theorem]{Corollary}

\newtheorem{definition}{Definition}

\theoremstyle{definition}

\theoremstyle{remark}

\newtheorem{myremark}{Remark} 
\newenvironment{remark}{\begin{myremark}}{$\Box$\end{myremark}}

\newtheorem{myexample}{Example}
\newenvironment{example}{\begin{myexample}}{\end{myexample}}

\else
\spnewtheorem{protocol}{Protocol}{\bfseries}{\rmfamily}
\spnewtheorem{algm}{Algorithm}{\bfseries}{\rmfamily}
\spnewtheorem{fact}{Fact}{\bfseries}{\rmfamily}
\spnewtheorem{myclaim}{Claim}{\bfseries}{\itshape}
\fi

\newcommand{\thmref}[1]{Theorem~\ref{thm:#1}}

\newcommand{\lemref}[1]{Lemma~\ref{lem:#1}}

\newcommand{\corref}[1]{Corollary~\ref{cor:#1}}

\newcommand{\secref}[1]{Section~\ref{sec:#1}}

\newcommand{\appref}[1]{Appendix~\ref{app:#1}}

\newcommand{\figref}[1]{Figure~\ref{fig:#1}}

\newcommand{\eqnref}[1]{equation~(\ref{eq:#1})}

\newcommand{\tr}{\mbox{Tr}}

\newcommand{\brak}[1]{{\langle {#1} \rangle}}
\newcommand{\set}[1]{\left\{ {#1} \right\}}
\newcommand{\paren}[1]{\left( {#1} \right)}
\newcommand{\sparen}[1]{\left[ {#1} \right]}

\newcommand{\floor}[1]{\left\lfloor{#1} \right\rfloor}

\DeclareMathOperator{\poly}{poly}

\newcommand{\E}{\mathbb{E}}

\newcommand{\I}{\mathbb{I}}

\newcommand{\N}{\mathbb{N}}

\newcommand{\p}{\mathsf{Pr}}

\newcommand{\R}{\mathbb{R}}
\newcommand{\bS}{\mathbb{S}}

\newcommand{\Y}{\mathbb{Y}}

\newcommand{\bPi}{\mathbf{\Pi}}
\newcommand{\bPhi}{\mathbf{\Phi}}
\newcommand{\bA}{\mathbf{A}}
\newcommand{\bB}{\mathbf{B}}
\newcommand{\bP}{\mathbf{P}}
\newcommand{\bW}{\mathbf{W}}
\newcommand{\bD}{\mathbf{D}}

\newcommand{\by}{\mathbf{y}}
\newcommand{\bx}{\mathbf{x}}
\newcommand{\bz}{\mathbf{z}}

\newcommand{\cD}{\mathcal{D}}
\newcommand{\cE}{\mathcal{E}}

\newcommand{\cH}{\mathcal{H}}

\newcommand{\cN}{\mathcal{N}}

\newcommand{\cQ}{\mathcal{Q}}

\newcommand{\cX}{\mathcal{X}}

\newcommand{\beq}{\begin{equation}}
\newcommand{\eeq}{\end{equation}}
\newcommand{\bml}{{\begin{multline}}}
\newcommand{\eml}{{\end{multline}}}

\usepackage[top=1in, bottom=1in, left=1in, right=1in]{geometry}

\excludecomment{versionA}
\includecomment{versionC}

\newcommand{\bU}{\mathbf{U}}
\newcommand{\bV}{\mathbf{V}}
\newcommand{\bLambda}{\mathbf{\Lambda}}
\newcommand{\bSigma}{\mathbf{\Sigma}}
\newcommand{\bE}{\mathbf{E}}

\excludecomment{versionB}
\begin{versionA}
\makeatletter
\renewcommand{\paragraph}{%
 \@startsection{paragraph}{4}%
 {\z@}{1ex \@plus 1ex \@minus .2ex}{-0.35em}%
{\normalfont\normalsize\bfseries}%
}
\makeatother

\makeatletter
\g@addto@macro \normalsize {%
 \setlength\abovedisplayskip{6pt plus 2pt minus 2pt}%
 \setlength\belowdisplayskip{6pt plus 2pt minus 2pt}%
}
\makeatother
\end{versionA}

\usepackage{titlesec}

\usepackage[skip=10pt]{caption}
\newcommand{\cov}{\mathsf{COV}}
\newcommand{\sber}{\mathsf{Rad}}
\newcommand{\PDF}{\mathsf{PDF}}

\begin{document}
\title{Randomness Efficient Fast-Johnson-Lindenstrauss Transform with Applications in Differential Privacy}
\author{Jalaj Upadhyay \\
Cheriton School of Computer Science \\
University of Waterloo. \\
\small{\sf jkupadhy@cs.uwaterloo.ca}}
\date{}
\maketitle

\begin{abstract}
\begin{versionB}
The Johnson-Lindenstrauss property ({\sf JLP}) of random matrices has immense application in computer science ranging from compressed sensing, learning theory, numerical linear algebra, to privacy. This paper explores the properties and applications of a distribution of random matrices. Our distribution satisfies {\sf JLP} with desirable properties like fast matrix-vector multiplication, sparsity, and optimal subspace embedding. We can sample a random matrix from this distribution using exactly $3n$ random bits. We show that a random matrix picked from this distribution preserves differential privacy under the condition that the input private matrix satisfies certain spectral property. This improves the run-time of various differentially private mechanisms like Blocki {\it et al.} (FOCS 2012) and Upadhyay (ASIACRYPT 13). Our final construction has a bounded column sparsity. Therefore, this answers an open problem stated in Blocki {\it et al.} (FOCS 2012). Using the results of Baranuik {\it et al.} (Constructive Approximation: 28(3)), our result implies a randomness efficient matrices that satisfies the Restricted-Isometry Property of optimal order for small sparsity with exactly linear random bits.  

We also show that other known distributions of sparse random matrices with the {\sf JLP} does not preserves differential privacy; thereby, answering one of the open problem posed by Blocki {\it et al.} (FOCS 2012). Extending on the works of Kane and Nelson (JACM: 61(1)), we also give unified analysis of some of the known Johnson-Lindenstrauss transform. We also present a self-contained simplified proof of an inequality on quadratic form of Gaussian variables that we use in all our proofs.
\end{versionB}

The Johnson-Lindenstrauss property ({\sf JLP}) of random matrices has immense applications in computer science ranging from compressed sensing, learning theory, numerical linear algebra, to privacy. This paper explores the properties and applications of a distribution of random matrices. Our distribution satisfies {\sf JLP} with desirable properties like fast matrix-vector multiplication, bounded sparsity, and optimal subspace embedding. We can sample a random matrix from this distribution using exactly $2n+n\log n$ random bits. We show that a random matrix picked from this distribution preserves differential privacy if the input private matrix satisfies certain spectral properties. This improves the run-time of various differentially private algorithms like Blocki {\it et al.} (FOCS 2012) and Upadhyay (ASIACRYPT 13). We also show that successive applications in a specific format of a random matrix picked from our distribution also preserve privacy, and, therefore, allows faster private low-rank approximation algorithm of Upadhyay (arXiv 1409.5414). Since our distribution has bounded column sparsity, this also answers an open problem stated in Blocki {\it et al.} (FOCS 2012). We also explore the application of our distribution in manifold learning, and give the first  differentially private algorithm for manifold learning if the manifold is smooth. Using the results of Baranuik {\it et al.} (Constructive Approximation: 28(3)), our result also implies a distribution of random matrices that satisfies the Restricted-Isometry Property of optimal order for small sparsity.  

We also show that other known distributions of sparse random matrices with the {\sf JLP} does not preserve differential privacy, thereby answering one of the open problem posed by Blocki {\it et al.} (FOCS 2012). Extending on the work of Kane and Nelson (JACM: 61(1)), we also give a unified analysis of some of the known Johnson-Lindenstrauss transform. We also present a self-contained simplified proof of the known inequality on quadratic form of Gaussian variables that we use in all our proofs. This could be of independent interest.

\medskip
\paragraph{Keywords.} 
Johnson-Lindenstrauss Transform, Differential Privacy, Restricted-Isometry Property.
\end{abstract}

\pagebreak

\pagenumbering{arabic}
\section{Introduction} \label{sec:introduction} The Johnson-Lindenstrauss lemma is a very useful tool for speeding up many high dimensional problems in computer science. Informally, it states that points in a high dimensional space can be embedded in low-dimensional space so that all pairwise distances are preserved up to a small multiplicative factor. Formally,
\begin{theorem} \label{thm:JL} 
Let $\varepsilon, \delta >0$ and $r=O(\varepsilon^{-2}\log (1/\delta))$. A distribution $\cD$ over ${r \times n}$ random matrix satisfies the Johnson-Lindenstrauss property {\em ({\sf JLP})} with parameters $(\varepsilon,\delta)$ if for $\bPhi \sim \cD$ and any $\bx \in \bS^{n-1}$, where $\bS^{n-1}$ is the $n$-dimensional unit sphere,
		\begin{align} \p_{\bPhi \sim \cD} \sparen{ \left| \sqrt{\frac{n}{r}}\|\bPhi \bx \|_2^2 -1 \right| > \varepsilon } < \delta.  \label{eq:JL} \end{align} 
\end{theorem}

Matrices with {\sf JLP} have found applications in various areas of computer science. For example, it has been used in numerical linear algebra~\cite{CW09,MM13,NN13,Sarlos06}, learning theory~\cite{BBV06},  quantum algorithms~\cite{CHTW04}, Functional analysis~\cite{JN10}, and compressed sensing~\cite{BDDW07}. Therefore, it is not surprising that many distribution over random matrices satisfying {\sf JLP} have been proposed that improves on various algorithmic aspects, like using binary valued random variables~\cite{Achlioptas03}, fast matrix-vector multiplication~\cite{AC09,Matousek08}, bounded sparsity~\cite{DKS10,KN14},  and de-randomization~\cite{KN10}. These improvements extends naturally to all the above listed applications.

Recently, Blocki {\it et al.}~\cite{BBDS12} found an application of {\sf JLP} in {\em differential privacy} (see, Definition~\ref{def:approxdp}). They showed that multiplying a private input matrix and a random matrix with i.i.d. Gaussian variables preserves differential privacy. Such matrices are known to satisfy {\sf JLP} for all values of $r$. However, the matrix multiplication is very slow. Subsequently, Upadhyay~\cite{Upadhyay13} showed that one could use graph sparsification in composition with the algorithm of Blocki {\it et al.}~\cite{BBDS12} to give a more efficient differentially private algorithm. Recently, Upadhyay~\cite{Upadhyay14} showed that non-private algorithms to solve linear algebraic tasks over streamed data can be made differentially private without significantly increasing the space requirement of the underlying data-structures. Unfortunately, these private algorithms update the underlying data-structure much slower than their non-private analogues~\cite{CW09,KN14}. 
 The efficient non-private algorithms use other distribution with {\sf JLP} that provides faster embedding. On the other hand, it is not known whether other distribution of random matrices with {\sf JLP}  preserve differential privacy or not.

This situation is particularly unfortunate when we consider differential privacy in the streaming model. In the streaming model, one usually receives an update which requires computation of $\bPhi \mathbf{e}_i$, where $\{\mathbf{e}_i\}_{i=1}^n$ are standard basis of $\R^n$. Since $\| \mathbf{e}_i \|=1$, even the naive approach takes $O(r \| \mathbf{e}_i \|_0)$ time, where $\|\cdot\|_0$ denotes the number of non-zero entries. This is faster than the current state-of-the-art privacy preserving techniques. For example, the  most efficient private technique using low-dimension embedding uses random matrices with every entries picked i.i.d. from a Gaussian distribution. Such matrices requires $\Omega(rn)$ time to update the privacy preserving data-structure even if we update  few entries. 

One of the primary difficulties in extending the privacy proof to other distributions that satisfies {\sf JLP} is that they introduce dependencies between random variables which might leak the sensitive information! Considering this, Blocki {\it et al.}~\cite{BBDS12} asked the question whether other known distribution of matrices with {\sf JLP} preserve differential privacy. 
In~\secref{notDP}, we answer this question in negative by showing that known distributions over sparse matrices like~\cite{AL09,DKS10,KN14,Matousek08,NPW14} fail to  preserve differential privacy (see,~\thmref{notDP} for precise statement). In the view of this  negative result, we  investigate a new distribution of random matrices. We show that careful composition of random matrices of appropriate form satisfies {\sf JLP}, allows fast matrix-vector multiplication, uses almost linear random bits, and preserves differential privacy.  We also achieve a reasonable range of parameters that allows other applications like compressed sensing, numerical linear algebra, and learning theory. 
We prove the following in~\secref{newJL}.
\begin{quote} {\bf Theorem} (Informal). For $\varepsilon,\delta >0$ and $r=O(\varepsilon^{-2} \log (1/\delta)) < n^{1/2}$, there exists a distribution $\cD$ over $r \times n$ matrices  that uses only almost linear random bits and satisfies {\sf JLP}. Moreover, for any vector $\bx \in \R^n$, $\bPhi \bx$ can be computed in $O\paren{n \log \paren{\min \set{\varepsilon^{-1} \log^4 (n/\delta),r}}}$ time for $\bPhi \sim \cD$.
\end{quote}
\begin{versionC}
\end{versionC}
\subsection{Contributions and Techniques} \label{sec:motivation}
The resources considered in this paper are randomness and time. The question of investigating {\sf JLP} with reduced number of random bits have been extensively studied in the non-private setting, most notably by Ailon and Liberty~\cite{AL09} and  Kane and Nelson~\cite{KN10}. In fact, Dasgupta {\it et al.}~\cite{DKS10} listed derandomizing their Johnson-Lindenstrauss transform as one of the major open problems.
Given the recent applications of matrices with {\sf JLP} to guarantee privacy in the streaming setting, it is important to investigate whether  random matrices that allows fast matrix-vector multiplication  preserve privacy or not. Recently, Mironov~\cite{Mironov12} showed a practical vulnerability in using random samples over continuous support. Therefore, it is also important to investigate if we can reduce the number of random samples in a differentially private algorithm.

\paragraph{A new distribution.} Our distribution consists of  a series of carefully chosen matrices such that all the matrices are formed using almost linear random samples, have few non-zero entries, and allow fast matrix-vector multiplication. Our construction has the form $\bPhi=\bP \bPi \mathbf{G}$, where $\mathbf{G}$ is a block diagonal matrices, $\bPi$ is a random permutation, and $\bP$ is sparse matrix formed by  vectors of dimension $n/r$. Each non-zero entries of $\bP$ are picked i.i.d. from a  sub-gaussian distribution. Every block of the matrix $\mathbf{G}$  has the form $\bW\bD$, where $\bD$ is a diagonal matrix with  non-zero entries picked i.i.d. to be $\pm 1$ with probability $1/2$ and $\bW$ is a Hadamard matrix. $\mathbf{G}$ has the property that for an input $\bx \in \bS^{n-1}$, $\by=\mathbf{G}\bx$ is a vector in $\bS^{n-1}$ with bounded coordinates. The permutation $\bPi$ helps to assure that if we break $\by$ into $r$-equal blocks, then every block has at least one non-zero entry with high probability. We could have used any other method present in the literature~\cite{decoupling}. We chose this method because it is simple to state and helps in proving both the concentration bound and differential privacy. 

\medskip \noindent {\sc Proof technique.} Our proof that the distribution satisfies {\sf JLP} uses a proof technique from cryptography. In a typical cryptographic proof, we are required to show that two probability ensembles $\cE_1$ and $\cE_2$ are indistinguishable. To do so, we often break the proof into series of hybrids $\cH_0=\cE_1, \cdots, \cH_{\mathsf p}=\cE_2$. We then prove that $\cH_i$ and $\cH_{i+1}$ are indistinguishable for $0 \leq i \leq p-1$. To prove our result, we perfectly emulate the action of a random matrix,  picked from our distribution, on a vector. Our emulation uses a series of matrix-vector multiplication. We then prove that a each of these matrix-vector multiplication preserves isometry under the conditions imposed by the matrices that have already operated on the input vector. 

\medskip \noindent {\sc Applications.} In this paper, we largely focus on two applications of the Johnson-Lindenstrauss lemma: compressed sensing and differential privacy. One of the main reasons why we concentrate on  these two applications of the Johnson-Lindenstrauss  is that these two applications have conflicting goals: in differential privacy, the aim is to conceal information about an individual data while in compressed sensing, one would like to decode an encoded sparse message. To date, the only distribution  that can be used in both these applications are the one in which entries of the matrix are picked i.i.d. from a Gaussian distribution. Unfortunately, such matrices allow very slow matrix-vector multiplication time and require a lot of random bits. 

\paragraph{Differential Privacy and {\sf JLP}.}
In the recent past, {\em differential privacy} has emerged as  a robust guarantee of privacy. The definition of differential privacy requires the notion of neighbouring dataset. We assume that datasets are represented in the form of matrices. We say two datasets represented in the form of matrices $\bA$ and $\widetilde{\bA}$  are {\em neighbouring} if $\bA- \widetilde{\bA}$ is a rank-$1$ matrix with Euclidean norm at most one. This notion was also used recently for linear algebraic tasks~\cite{HR12,HR13,Upadhyay14} and for statistical queries~\cite{BBDS12,Upadhyay13}. The definition of privacy we use is given below. 
\begin{definition} \label{def:approxdp}
	A randomized algorithm, $\mathfrak{M}$, gives {\em $(\alpha, \beta)$-differential privacy}, if for all  
neighbouring matrices $\bA$ and $\widetilde{\bA}$ and all subset $S$ in the range of  $\mathfrak{M}$, 
		$ \p[\mathfrak{M}(\bA) \in S] \leq \exp (\alpha) \p[\mathfrak{M}(\widetilde{\bA}) \in S] + \beta, $ where the probability is over the coin tosses of $\mathfrak{M}$. 
\end{definition}

Blocki {\it et al.}~\cite{BBDS12}  showed that the Johnson-Lindenstrauss transform instantiated with random Gaussian matrices  preserve {differential privacy} for certain queries (note that, to answer all possible queries, they require $r \geq n$). Later, Upadhyay~\cite{Upadhyay14} approached this in a principled manner showing that random Gaussian matrices (even if it is applied twice in a particular way) preserves differential privacy. However,  the question whether other random matrices with {\sf JLP} preserve differential privacy or not remained open. 
In~\secref{impossible}, we show that known  distributions over sparse matrices with {\sf JLP} does not preserves differential privacy.

We give an example by studying the construction of Nelson {\it et al.}~\cite{NPW14}. 
Their construction requires a distribution of $R \times n$ random matrices that satisfies {\sf JLP} for suboptimal value of $R$. Let $\bPhi_{\mathsf{sub}} \bD$ be such a matrix, where $\bD$ is a diagonal matrix with non-zero entries picked i.i.d. to be $\pm 1$ with probability $1/2$. They form  an $r \times n$ matrix $\bPhi$ from $\bPhi_{\mathsf{sub}}$ by forming a random linear combination of the rows of $\bPhi_{\mathsf{sub}}$, where the coefficients of the linear combination is $\pm 1$ with probability $1/2$. In example~\ref{e.g.:fourier} below, we give a counter-example when $\bPhi_{\mathsf{sub}}$ is a subsampled Hadamard matrix to show that it fails to preserve differential privacy. Note that $\bPhi_{\mathsf{sub}} \bD$ is known to satisfy {\sf JLP} sub-optimally~\cite{CGV13,KW11}.
The counter-example when $\bPhi_{\mathsf{sub}}$ are partial circulant matrices follows the same idea using the observation that a circulant matrix formed by a vector $\mathbf{g}$ has the form $\mathbf{F}_n^{-1} \mathsf{Diag}(\mathbf{F}_n \mathbf{g}) \mathbf{F}_n$, where $\mathbf{F}_n$ is an $n \times n$ discrete Fourier transform.

\begin{example} \label{e.g.:fourier}
Let $\bPhi_{\mathsf{sub}}$ be a matrix formed by independently sampling $R$ rows of a Hadamard matrix. 
Let $\bPhi \bD$ be the matrix with {\sf JLP} as guaranteed by Nelson {\it et al.}~\cite{NPW14} formed by random linear combination of the rows of $\bPhi_{\mathsf{sub}}$. We give two $n \times n$ matrices $\bA$ and $\widetilde{\bA}$ such that their output distributions $\bPhi \bD \bA $ and $\bPhi \bD \widetilde{\bA}$ are easily distinguishable. Let $\mathbf{v}=\begin{pmatrix} 1 & 0 & \cdots & 0 \end{pmatrix}^{\mathsf T}$, then 
\begin{versionC}
\[
\bA= 	\begin{pmatrix} w & 1 & 0 & \cdots & 0 \\  
				0 & w & 0 & \cdots & 0 \\
  				0  & 0 & w & \cdots & 0\\
  				\vdots & \vdots & & \ddots& \vdots \\
  				{0} & 0 & 0 & \cdots & w
			\end{pmatrix} \quad	\text{and} \quad
\widetilde{\bA}= 	\begin{pmatrix} w & 1 & 0 & \cdots & 0 \\  
				1 & w & 0 & \cdots & 0 \\
  				0  & 0 & w & \cdots & 0\\
  				\vdots & \vdots & & \ddots &\vdots \\
  				{0} & 0 & 0 & \cdots & w
			\end{pmatrix} 	
		 \]  
\end{versionC}
 \begin{versionA}
\[
\bA= w\I + \mathbf{v} \mathbf{e}_2^{\mathsf T} \quad \text{and} \quad
\widetilde{\bA}=w\I + \mathbf{v} \mathbf{e}_2^{\mathsf T} + \mathbf{e}_2 \mathbf{v}^{\mathsf T}
 \]  
\end{versionA}
  \begin{versionB}
\[
\bA= 	\begin{pmatrix} w & 1 & \mathbf{0} \\  
				0 & w & \mathbf{0} \\
  				\mathbf{0}  & \mathbf{0} & w\I
			\end{pmatrix} \quad	\text{and} \quad
\widetilde{\bA}=  	\begin{pmatrix} w & 1 & \mathbf{0} \\  
				1 & w & \mathbf{0} \\
  				\mathbf{0}  & \mathbf{0} & w\I
			\end{pmatrix}
		 \]  
\end{versionB}
for $w > 0$ and identity matrix $\I$. It is easy to verify that $\bA$ and $\widetilde{\bA}$ are neighbouring. Now, if we concentrate on all possible left-most top-most $2 \times 2$ sub-matrices formed by multiplying $\bPhi$, simple matrix multiplications shows that with probability $1/\poly(n)$, the output distributions when $\bA$ and $\widetilde{\bA}$ are used are easily distinguishable. This violates differential privacy. 
\begin{versionA}
We give the details in the full version. Note that for large enough $w$,   known results~\cite{BBDS12,Upadhyay14} shows that using a random Gaussian matrix instead of $\bPhi \bD$ preserves privacy. 
\end{versionA}
\begin{versionC}
We defer the details of this calculation to~\secref{NPW14}. Note that for large enough $w$,   known results~\cite{BBDS12,Upadhyay14} shows that using a random Gaussian matrix instead of $\bPhi \bD$ preserves privacy. 

\end{versionC}
\end{example}

Example~\ref{e.g.:fourier} along with the examples in~\secref{notDP} raises the following  question: 	{\em Is there a distribution over sparse matrices which has optimal-{\sf JLP}, allows fast matrix-vector multiplication, and preserves differential privacy?} We answer this question affirmatively in~\secref{DP}.  
 Informally, 
\begin{quote} {\bf Theorem} (Informal). If the singular values of an $m \times n$ input matrix $\bA$ are at least $\ln(4/\beta) \sqrt{16 r \log (2/\beta)}/{\alpha}$, then there is a distribution $\cD$ over random matrices such that $\bPhi \bA^{\mathsf T}$ for $\bPhi \sim \cD$ preserves $(\alpha, \beta+\mathsf{negl}(n))$-differential privacy.  Moreover, the running time to compute $ \bPhi \bA^{\mathsf T}$ is $O(mn \log r)$.
\end{quote}

 \begin{table} [t]
{\small
{
\begin{tabular}{|c|c|c|c|c|}
\hline
Method & Cut-queries & Covariance-Queries & Run-time & $\#$ Random Samples \\ \hline
Exponential Mechanism~\cite{BLR08,MT07}& $O(n \log n /\alpha)$ & $O(n \log n /\alpha)$ & Intractable & $n^2$\\ \hline
Randomized Response~\cite{GRU12} & $O(\sqrt{n |S| \log |\cQ |/\alpha})$ & $\widetilde{O}(\sqrt{d \log |\cQ |}/\alpha)$ & \textcolor{blue}{$\Theta(n^2)$} & $n^2$ \\ \hline
Multiplicative Weight~\cite{HR10} & $\widetilde{O} (\sqrt{|\cE|} \log |\cQ | /\alpha)$  & $\widetilde{O}(d \sqrt{n \log |\cQ |/\alpha})$ & $O(n^2)$ & $n^2$ \\ \hline
Johnson-Lindenstrauss~\cite{BBDS12} & \textcolor{blue}{$O(|S| \sqrt{\log |\cQ |}/\alpha)$} & \textcolor{blue}{$O(\alpha^{-2}\log |\cQ |)$} & $O(n^{2.38})$ & $rn$ \\ \hline
Graph Sparsification~\cite{Upadhyay13} & \textcolor{blue}{$O(|S| \sqrt{\log |\cQ |}/\alpha)$} & $-$ & $O(n^{2+o(1)})$ & $n^2$ \\ \hline
This paper &  \textcolor{blue}{$O(|S|  \sqrt{\log |\cQ |}/\alpha)$} & \textcolor{blue}{$O(\alpha^{-2}\log |\cQ |)$} & \textcolor{blue}{$O(n^2 \log b)$} &  \textcolor{blue}{$2n+n\log n$} \\ \hline
\end{tabular} \label{table:compare}
\caption{Comparison Between Differentially Private Mechanisms.}  \label{table}
}}\end{table}

\noindent {\sc Proof Technique.} To prove this theorem, we  reduce it to the proof of Blocki {\it et al.}~\cite{BBDS12} as done by Upadhyay~\cite{Upadhyay13}. However, there are subtle reasons why the same  proof  does not apply here. 
One of the reasons is that one of the sparse random matrix has many entries picked deterministically to be zero. This allows an attack if we try to emulate the proof of~\cite{BBDS12}. Therefore, we need to study the matrix-variate distribution imposed by the other matrices involved in our projection matrix. In fact, we show a simple attack in~\secref{DP} if one try to prove privacy using the idea of Upadhyay~\cite{Upadhyay13} in our setting. 

\medskip \noindent {\scshape Comparison and private manifold learning.} Our result improves the efficiency of the algorithms of Blocki {\it et al.}~\cite{BBDS12} Upadhyay~\cite{Upadhyay13}, and the update time of the differentially private streaming algorithms of Upadhyay~\cite{Upadhyay14}. In Table~\ref{table}, we compare our result with previously known bounds to answer cut-queries on an $n$-vertex graph and covariance queries on an $n \times m$ matrix. In the table, $S$ denotes the subset of vertices for which the analyst makes the queries, $\cQ$ denotes the set of queries an analyst makes, $\cE$ denotes the set of edges of the input graph, $ b= \min \{c_2 a \log\paren{\frac{r}{\delta}},r\}$ with $a=c_1 \varepsilon^{-1} \log\paren{\frac{1}{\delta}} \log^2\paren{\frac{r}{\delta}}$. The second and third columns represents the additive error incurred in the algorithm for differential privacy. Table~\ref{table} shows that we achieve the same error bound as in~\cite{BBDS12, Upadhyay13}, but with improved run-time (note that $b \ll n$) and  quadratically less random bits.
\begin{versionA}
We also initiate the study of differentially private manifold learning. We give an efficient, non-interactive private algorithm under certain smoothness condition on the manifold (see,~\thmref{manifold}).
\end{versionA} 

\begin{versionC}
Our theorem for differential privacy is extremely flexible. In this paper, we exhibit its flexibility by showing a way to convert a non-private algorithm for manifold learning to a private learning algorithm. In manifold learning, we are $m$ points  $\bx_1,\cdots, \bx_m \in \R^n$ that lie on an $n$-dimensional manifold $\mathscr{M}$. We have the guarantee that it can be described by $f:\mathscr{M} \rightarrow \R^n$. The goal of the learning algorithm is to find $\by_1, \cdots, \by_m$ such that $\by_i=f(\bx_i)$. It is often considered non-linear analogs to principal component analysis. We give a non-adaptive differentially private algorithm under certain smoothness condition on the manifold.  
\end{versionC}

\medskip \noindent {\scshape Private low-rank approximation.} The above theorem speeds up all the known private algorithms based on the Johnson-Lindenstrauss lemma, except for the differentially-private low-rank approximation ($\mathsf{LRA}$)~\cite{Upadhyay14}. This is because the privacy proof in~\cite{Upadhyay14} depends crucially on the fact that two applications of dense Gaussian random matrices in a specific form also preserves privacy. To speed up the differentially private  $\mathsf{LRA}$ of~\cite{Upadhyay14}, we need to prove analogous result. In general, reusing random samples can lead to privacy loss. Another complication is that there are fixed non-zero entries in $\bP$. In~\thmref{DP2}, we prove that differential privacy is still preserved with our distribution. 

Our bound shows a gradual decay of the run-time as the sparsity increases. Note that the bounds of~\cite{NPW14} is of interest only for $r \geq \sqrt{n}$; it is sub-optimal for the regime considered by~\cite{AL09,AR14}. 


\begin{versionC}
\paragraph{Other Applications.} Due to the usefulness of low-dimension embedding, there have been considerable efforts by researchers to prove such a dimensionality reduction theorem in other normed spaces. All of these efforts have thus far resulted in negative results which show that the {\sf JLP} fails to hold true in certain non-Hilbertian settings. More specifically, Charikar and Sahai~\cite{CS02} proved that there is no dimension reduction via linear mappings in $\ell_1$-norm. This was extended to the general case by Johnson and Naor~\cite{JN10}. They showed that a normed space that satisfies the {\sf JL} transform is very close to being Euclidean in the sense that all its subspaces are close to being isomorphic to Hilbert space. 

Recently, Allen-Zhu, Gelashvili, and Razenshteyn~\cite{AGR14} showed that there does not exists a projection matrix with $\mathsf{RIP}$ of optimal order for any $\ell_p$ norm other than $p=2$. In the view of Krahmer and Ward~\cite{KW11} and characterization by Johnson and Naor~\cite{JN10}, our result gives an alternate reasoning for their impossibility result. Moreover, using the result by Pisier~\cite{Pisier}, which states that, if a Banach space has Banach-Mazur distance  $o(\log n)$, then the Banach space is super-reflexive, we get one such space using just almost linear samples. This has many applications in Functional analysis (see~\cite{Conway} for details).
\end{versionC}

\paragraph{Other Contributions.}
In~\secref{impossible},  {\em we answer negatively to the question raised by Blocki {\it et al.}~\cite{BBDS12} whether already known distributions of random matrices with {\sf JLP} preserve differential privacy or not}. We show this by giving counterexamples of neighbouring datasets on which the known distributions fail to preserve differential privacy with constant probability. Furthermore, extending on the efforts of Kane and Nelson~\cite{KN10,KN14} and Dasgupta and Gupta~\cite{DG03}, in~\secref{older}, we use~\thmref{hw} to  give a unified and simplified analysis of some of the  known distributions of random matrices that satisfy the {\sf JLP}. We also give a simplified proof of~\thmref{hw} with Gaussian random variables in the full version.
\begin{versionA}
Our result also has applications in functional analysis. Johnson and Naor~\cite{JN10} have shown that a normed space that satisfies {\sf JLP} is very close to being Euclidean in the sense that all its subspaces are close to being isomorphic to Hilbert space. Recently, Allen-Zhu {\it et al.}~\cite{AGR14} showed that  an optimal $\mathsf{RIP}$ matrix for any $\ell_p$ norm other than $p=2$  does not exist. Combining Krahmer and Ward~\cite{KW11} with~\cite{JN10}, our result gives an alternate reasoning for their impossibility result.  Using the result by Pisier~\cite{Pisier}, our theorem gives a super-reflexive Banach space using $2n+n\log n$ random bits. This has many applications in Functional analysis (see~\cite{Conway} for details). 
\end{versionA}

\subsection{Related Works}
In what follows, we just state the result mentioned in the respective papers. Due to the relation between {\sf JLP} and {\sf RIP}, the corresponding result also follows.
The original {\sf JL} transform used a distribution over random Gaussian matrices~\cite{JL84}. 
Achlioptas~\cite{Achlioptas03} observed that we only need $(\bPhi_{i:} \bx)^2$ to be concentrated around the mean for all unit vectors, and therefore, entries picked i.i.d. to $\pm 1$ with probability $1/2$ also satisfies {\sf JLP}. 

These two distributions  are over dense random matrices with slow embedding time. From application point of view, it is important to have fast embedding. 
An obvious way to improve the efficiency of the embedding is to use distribution over sparse random matrices. 
The first major breakthrough in this came through the work of Dasgupta, Kumar, and Sarlos~\cite{DKS10}, who gave a optimal distribution over random matrices with only ${O}(\varepsilon^{-1} \log ^3(1/\delta))$ non-zero entries in every column. Subsequently, Kane and Nelson~\cite{KN10, KN14} improved it to $\Theta(\varepsilon^{-1} \log (1/\delta))$ non-zero entries in every column.
Ailon and Chazelle~\cite{AC09} took an alternative approach. They used composition of matrices, each of which allows fast embedding.  However, the amount of randomness involved to perform the mapping was $O(rn)$. This was improved  by Ailon and Liberty~\cite{AL09}, who used dual binary codes to a distribution of random matrices using almost linear random samples.

The connection between {\sf RIP} and {\sf JLP} offers another approach to improve the embedding time. 
\begin{versionA}
In the large regime of $r$, the best known construction of {\sf RIP} matrices with fast matrix-vector multiplication are based on bounded orthonormal matrices or partial circulant matrices based constructions. Both these constructions are sub-optimal by three logarithmic factors, even through the improvements by Nelson {\it et al.}~\cite{NPW14}.  The case when $r \ll \sqrt{n}$ has been solved optimally in~\cite{AL09}, with the recent improvement by~\cite{AL13}. 
However, the number of random bits they require grows linearly with the number of iteration of $\bW\bD$ they perform. Table~\ref{tableRIP} contains a concise description of these results.
\end{versionA}
\begin{versionC}
In the large regime of $r$, there are known distributions that satisfies {\sf RIP} optimally, but all such distributions are over dense random matrices and permits slow multiplication. The best known construction of {\sf RIP} matrices  for large $r$ with fast matrix-vector multiplication are based on bounded orthonormal matrices or partial circulant matrices based constructions. Both these constructions are sub-optimal by several logarithmic factors, even through the improvements by Nelson {\it et al.}~\cite{NPW14}, who gave an ingenuous technique to bring down this to just three factors off from optimal.  The case when $r \ll \sqrt{n}$ has been solved optimally in~\cite{AL09}, with the recent improvement by~\cite{AL13}. They showed that iterated application of $WD$ matrices followed by sub-sampling gives the optimal bound. However, the number of random bits they require grows linearly with the number of iteration of $\bW \bD$ they perform. Table~\ref{tableRIP} contains a concise description of these results.
\end{versionC}
There is a rich literature on differential privacy. The formal definition was  given by Dwork {\it et al.}~\cite{DMNS06}. They used Laplacian distribution to guarantee differential privacy for bounded {\em sensitivity} query functions. The Gaussian variant of this basic algorithm was shown to preserve differential privacy by 
Dwork {\it et al.}~\cite{DKMMN06} in a follow-up work. Since then, many mechanisms for preserving differential privacy have been proposed in the literature (see, Dwork and Roth for detailed discussion~\cite{DR14}).  All these mechanisms has one 
common feature: they perturb the output before responding to  queries. Blocki {\it et al.}~\cite{BBDS12, BBDS13} and Upadhyay~\cite{Upadhyay13,Upadhyay14} took a complementary approach: they perturb the matrix  and perform a random projection on this matrix. Blocki {\it et al.}~\cite{BBDS12} and Upadhyay~\cite{Upadhyay13} used this idea to show that one can answer cut-queries on a graph, while Upadhyay~\cite{Upadhyay14} showed that if the singular values of the input private matrix is high enough, then even two successive applications of random Gaussian matrices in a specified format preserves privacy.

\medskip \noindent{\scshape Organization of the paper.} \secref{prelims} covers the basic preliminaries and notations,~\secref{newJL} covers the new construction and its proof,~\secref{applications} covers its applications. We prove some of the earlier known constructions does not preserve differential privacy in~\secref{impossible}. 
\begin{versionA}
\end{versionA}
\begin{versionC}
We conclude this paper in~\secref{open}
\end{versionC}

\section{Preliminaries and Notations} \label{sec:prelims} We  denote a random matrix by $\bPhi$. We use  $n$ to denote the dimension of the unit vectors,  $r$ to denote the dimension of the embedding subspace. We fix $\varepsilon$ to denote the accuracy parameter and $\delta$ to denote the confidence parameter in the {\sf JLP}. We use bold faced capital letters like $\bA$ for matrices, bold face small letters to denote vectors like $\bx, \by$, and calligraphic letters like $\cD$ for distributions. For a matrix $\bA$, we denote its Frobenius norm by $\|\bA\|_F$ and its operator norm by $\|\bA\|$. We denote a random vector with each entries picked i.i.d. from a sub-gaussian distribution by  $\mathbf{g}$. We use $\mathbf{e}_1, \cdots, \mathbf{e}_n$ to denote the standard basis of an $n$-dimensional space. We use the symbol $\mathbf{0}^n$ to denote an $n$-dimensional zero vector. We use $c$ (with numerical subscript) to denote universal constants.

\noindent {\bf Subgaussian random variables.} A random variable $X$ has a {\em subgaussian distribution}
with scale factor $\lambda <\infty$ if  for all real $t$, $\p [e^{t X}] \leq e^{\lambda^2 t^2/2}$. We use the following subgaussian distributions:
\begin{versionA}
(i) {\em Gaussian distribution:} \label{defn:gaussian} A standard gaussian distribution $X \sim \cN(\mu,\sigma^2)$ is distribution with probability density function defined as $\frac{1}{\sqrt{2\pi \sigma}}e^{-(X-\mu)^2/2\sigma^2}$,
(ii) {\em Rademacher distribution:} \label{defn:rademacher} A random variable $X$ with distribution $\p [X=-1] = \p [X=1]=1/2$ is called a Rademacher distribution. We denote it by $\sber(1/2)$. 
\end{versionA}

\begin{versionC}
\begin{enumerate}
	\item {\bf Gaussian distribution:} \label{defn:gaussian} A standard gaussian distribution $X \sim \cN(\mu,\sigma)$ is distribution with probability density function defined as $\frac{1}{\sqrt{2\bPi \sigma}}\exp(-|X-\mu|^2/2\sigma^2)$.
	\item {\bf Rademacher distribution:} \label{defn:rademacher} A random variable $X$ with distribution $\p [X=-1] = \p [X=1]=1/2$ is called a Rademacher distribution. We denote it by $\sber(1/2)$.
	\item {\bf Bounded distribution:} \label{defn:bounded} A more general case is bounded random variable $X$ with the property that $|X| \leq M$ almost surely for some constant $M$. Then $X$ is sub-gaussian.
\end{enumerate}
\end{versionC}

A  Hadamard matrix of order $n$ is defined recursively as follows: $ \bW_n = \frac{1}{\sqrt{2}}\begin{pmatrix}  \bW_{n/2} & \bW_{n/2} \\ \bW_{n/2} & -\bW_{n/2} \end{pmatrix}, \bW_1=1.\label{eq:hadamard} $ When it is clear from the context, we drop the subscript. We call a matrix randomized  Hadamard matrix of order $n$ if it is of the form $ \bW \bD$, where $\bW$ is a normalized  Hadamard matrix and $\bD$ is a diagonal matrix with non-zero entries picked i.i.d. from $\sber(1/2).$  We fix $\bW$ and $\bD$ to denote these matrices.

\begin{versionC}
We will few well known concentration bounds for random subgaussian variables.
\begin{theorem} {\em (Chernoff-Hoeffding's inequality)} \label{thm:ch}
Let $X_1, \cdots, X_n$ be $n$ independent  random variables with the same probability distribution, each ranging over the real interval $[i_1, i_2]$, and let $\mu$ denote the expected value of each of these variables. Then, for every $\varepsilon >0$, we have
\[ \p \sparen{ \left| \frac{\sum_{i=1}^n X_i}{n} - \mu \right| \geq \varepsilon } \leq 2 \exp \paren{- \frac{2\varepsilon^2n}{(i_2-i_1)^2} }.  \]
\end{theorem} 
The Chernoff-Hoeffding bound is useful for estimating the average value of a function defined over a large set of values, and this is exactly where we use it. Hoeffding showed that this theorem holds for sampling without replacement as well. However, when the sampling is done without replacement, there is an even sharper bound by Serfling~\cite{Serfling74}.
\begin{theorem} {\em (Serfling's bound)} \label{thm:serfling}
Let $X_1, \cdots, X_n$ be $n$ sample drawn without replacement from a list of $N$ values, each ranging over the real interval $[i_1, i_2]$, and let $\mu$ denote the expected value of each of these variables. Then, for every $\varepsilon >0$, we have
\[ \p \sparen{ \left| \frac{\sum_{i=1}^n X_i}{n} - \mu \right| \geq \varepsilon } \leq 2 \exp \paren{- \frac{2\varepsilon^2n}{(1-(n-1)/N)(i_2-i_1)^2} }.  \]
\end{theorem} 
The quantity $1-(n-1)/N$ is often represented by $f^*$. Note that when $f^*=0$, we derive~\thmref{ch}. Since $f^* \leq 1$, it is easy to see that~\thmref{serfling} achieves better bound than~\thmref{ch}.

We also use the following result of Hanson-Wright~\cite{HW} for subgaussian random variables.
\begin{theorem} {\em (Hanson-Wright's inequality).} \label{thm:hw}
	For any symmetric matrix $\bA$, sub-gaussian random vector $\bx$ with mean $0$ and variance $1$, 
	\begin{align} \p_{\mathbf{g}} \sparen{ | \mathbf{g}^{\mathsf T} \bA \mathbf{g} - \tr(\bA) | > \eta } \leq 2 \exp \paren{ -\min \set{ {c_1 \eta^2}/{\|\bA\|_F^2} , {c_2 \eta}/{\|\bA\|}} }, \label{eq:hw} \end{align}
	where $\|A\|$ is the operator norm of matrix $\bA$. 
\end{theorem}
The original proof of Hanson-Wright inequality is very simple when we just consider special class of subgaussian random variables. We give the proof when $\mathbf{g}$ is a random Gaussian vector in~\appref{gaussian}. A proof for $\sber(1/2)$ is also quite elementary. For a simplified proof for any subgaussian random variables, we refer the readers to~\cite{RV}.

 In our proof of differential privacy, we prove that each row of the published matrix preserves $(\alpha_0, \beta_0)$-differential privacy for some appropriate $\alpha_0,\beta_0$, and then invoke a composition theorem proved by Dwork, Rothblum, and Vadhan~\cite{DRV10} to prove that the published matrix preserves $(\alpha,\beta)$-differential privacy. The following theorem is the composition theorem that we use.
\begin{theorem} \label{thm:DRV10}
	Let $\alpha, \beta \in (0,1)$, and $\beta'>0$. If $\mathfrak{M}_1, \cdots , \mathfrak{M}_\ell$ are each $(\alpha, \beta)$-differential private mechanism, then the mechanism $\mathfrak{M}(D):= (\mathfrak{M}_1(D), \cdots , \mathfrak{M}_\ell(D))$ releasing the concatenation of each algorithm is $(\alpha', \ell \beta+\beta')$-differentially private for $\alpha' < \sqrt{2\ell \ln (1/\beta')}\alpha + 2\ell \alpha^2$.
\end{theorem}
\end{versionC}

\section{A New Distribution over Random Matrices} \label{sec:newJL}  In this section, we give our basic distribution of random matrices in~\figref{basic} that satisfies {\sf JLP}. We then perform subsequent improvements to reduce the number of random bits in~\secref{random} and improve the run-time in~\secref{run}. 
Our proof uses an idea from cryptography and an approach taken by Kane and Nelson~\cite{KN10,KN14} to use Hanson-Wright inequality. 
Unfortunately, we cannot just use the Hanson-Wright inequality as done by~\cite{KN14} due to the way our distribution is defined. This is because most of the entries of $\bP$ are non-zero and we use a  Hadamard and a permutation matrix in between the two matrices with sub-gaussian entries. To resolve this issue, we use a trick from cryptography. We emulate the action of a random matrix picked from our defined distribution by a series of random matrices. We show that each of these composed random matrices satisfies the {\sf JLP} under the constraints imposed by the other matrices that have already operated on the input.
The main result proved in this section is the following theorem.
\begin{theorem}  \label{thm:newsparse}
For any $\varepsilon, \delta >0$ and a positive integer $n$. Let $r=O(\varepsilon^{-2} \log (1/\delta))$, $a=c_1 \varepsilon^{-1} \log(1/\delta) \log^2(r/\delta)$, and $ b=c_2 a \log(a/\delta)$ with $r \leq n^{1/2-\tau}$ for arbitrary constant $\tau>0$. Then there exists a distribution $\cD$ of  $r \times n$ matrices over $\R^{r \times n}$ such that any sample $\bPhi \sim \cD$  needs only $2n+n\log n$ random bits and for any $\bx \in \bS^{n-1}$, 
		\begin{align} \p_{\bPhi \sim \cD} \sparen{ \left|~\|\bPhi \bx \|_2^2 -1 \right| > \varepsilon } < \delta.  \end{align}
		Moreover, the run-time of a matrix-vector multiplication takes $O(n \min\{\log b,\log r\})$ time. \end{theorem}
\begin{versionC}
\begin{figure} [t]
\begin{center}
\fbox{
\begin{minipage}[l]{6in}
\medskip
{\bf Construction of $\bPhi$:} Construct the matrices $\bD$ and $\bP$ as below.
\begin{enumerate}
	\item $\bD$ is an $n \times n$ diagonal matrix such that $\bD_{ii}\sim \sber(1/2)$ for $1 \leq i \leq n.$
  	\item $\bP$ is constructed as follows:
  	\begin{enumerate}
  		\item Pick $n$ random sub-gaussian samples $\mathbf{g}=\set{\mathbf{g}_1, \cdots, \mathbf{g}_n}$ of mean $0$ and variance $1$.
   		\item Divide $\mathbf{g}$ into blocks of  $n/r$ elements, 
  		such that for $1 \leq i \leq r$, $\bPhi_{i} = \paren{\mathbf{g}_{(i-1)n/r+1}, \cdots, \mathbf{g}_{ir}}$.
  		\item Construct the matrix $\bP$ as follows:
  		\begin{align} \begin{pmatrix}
  				\bPhi_1 & \mathbf{0}^{n/r} & \cdots &  \mathbf{0}^{n/r} \\
  				\mathbf{0}^{n/r} &\bPhi_2 & \cdots & \mathbf{0}^{n/r} \\
  				\vdots & \ddots & \ddots & \vdots \\
  				\mathbf{0}^{n/r} & \cdots & \mathbf{0}^{n/r} & \bPhi_r 
			\end{pmatrix} \nonumber
  		 \end{align}
  	\end{enumerate}
\end{enumerate}
Compute $\bPhi= \sqrt{\frac{1}{r}}\bP \bPi  \bW \bD$, where $\bW$ is a normalized  Hadamard matrix and $\bPi$ is a permutation matrix on $n$ entries.
\end{minipage}
}\caption{New Construction} \label{fig:basic} 
\end{center}
\end{figure}
\end{versionC}

\begin{versionA}
\begin{figure} [t]
\begin{center}
\fbox{
\begin{minipage}[l]{6in}
\medskip
\paragraph{Construction of the distribution.} The random matrix $\bPhi= \sqrt{\frac{1}{r}}\bP \bPi  \bW \bD$, where $\bW \bD$ is a normalized and randomized Hadamard matrix and $\bPi$ is a permutation matrix on $n$ entries. The matrix $\bP$ is constructed as follows: 
\begin{enumerate}
  		\item Pick $n$ random sub-gaussian samples $\mathbf{g}=\paren{\mathbf{g}_1, \cdots, \mathbf{g}_n}$ of mean $0$ and variance $1$.
  		\item Divide $\mathbf{g}$ into blocks of  $n/r$ elements, 
  		such that for $1 \leq i \leq r$, $\bPhi_{i} = \paren{\mathbf{g}_{(i-1)n/r+1}, \cdots, \mathbf{g}_{ir}}$.
  		\item Construct the matrix $\bP$ as follows:
  		\begin{align} \begin{pmatrix}
  				\bPhi_1 & \mathbf{0}^{n/r} & \cdots &  \mathbf{0}^{n/r} \\
  				\mathbf{0}^{n/r} &\bPhi_2 & \cdots & \mathbf{0}^{n/r} \\
  				\vdots & \ddots & \ddots & \vdots \\
  				\mathbf{0}^{n/r} & \cdots & \mathbf{0}^{n/r} & \bPhi_r 
			\end{pmatrix} \nonumber
  		 \end{align}
  	\end{enumerate}
\end{minipage}
}\caption{A New Distribution Over Random Matrices in $\R^{r \times n}$} \label{fig:basic} 
\end{center}
\end{figure}
\end{versionA}
We prove~\thmref{newsparse} by proving that a series of operations on $\bx \in \bS^{n-1}$ preserves the Euclidean norm. For this, it is helpful to analyze the actions of every matrix involved in the distribution from which $\bPhi$ is sampled. Let $\by=\bW \bD \bx$  and $\by'=\bPi \by$. We first look at the vector $\bP \by'$. 

 In order to analyze the action of $\bP$ on an input $\by' $, we consider an equivalent way of looking at the action of $\bP \bPi$ on $\by =  \bW \bD \bx$. In other words, we emulate the action of $\bP$ using $\bP_1$ and $\bP_2$, each of which are much easier to analyze. 
The row-$i$ of the matrix $\bP_1$ is formed by $\bPhi_i$ for $1 \leq i \leq r$. 
For an input vector $\by' = \bPi \by$, $\bP_2$ acts as follows: for any $1 \leq j \leq r$, $\bP_2$  samples $t= n/r$ columns entries $\paren{\by'_{(j-1)n/r+1}, \cdots, \by'_{jn/r}}$, multiply every entry by $\sqrt{n/t}$, and feeds it to the $j$-th row of $\bP_1$. In other words, $\bP_2$ feeds $\bz_{(j)} = \sqrt{n/t}\paren{\by'_{(j-1)n/r+1}, \cdots, \by'_{jn/r}}$ to the $j$-th row of $\bP_1$.

$\| \bW \bD \bx \|_2 =1$ as $\bW \bD$ is a unitary. We prove that  $\bP_1$ is an isometry in~\lemref{gaussian} and $\bP_2\bPi $ is an isometry in~\lemref{P_2}. The final result~(\thmref{newgaussian}) is obtained by combining  these two lemmata.


\begin{versionC}
We need the following result proved by Ailon and Chazelle~\cite{AC09}.
\begin{theorem} \label{thm:AC09}
Let $\bW$ be a $n \times n$  Hadamard matrix and $\bD$ be a diagonal signed Bernoulli matrix with $\bD_{ii} \sim \sber(1/2).$ Then for any vector $\bx \in \bS^{n-1}$, we have 
\[ \p_D \sparen{\| \bW \bD\bx\|_\infty \geq \sqrt{\paren{\log (n/\delta)}/{n}}} \leq \delta. \]
\end{theorem}
\end{versionC}

We use~\thmref{serfling} and~\thmref{AC09} to prove that $ \bP_2 \bPi$ is an isometry for $r \leq n^{1/2-\tau}$, where $\tau>0$ is  an arbitrary constant.~\thmref{AC09} gives us that $\sqrt{n/t} \| \bW\bD \bx\|_\infty \leq \sqrt{\log (n/\delta)/t}$. We use~\thmref{serfling} because $\bP_2\bPi$ picks $t$ entries of $\by$ without replacement. 
\begin{versionA} 
In the full version, we show the following:
\end{versionA}
\begin{lemma} \label{lem:P_2}
Let $\bP_2$ be a matrix 
as defined above. Let $\bPi$ be a permutation matrix. Then for any $\by \in \bS^{n-1}$ with $\| \by \|_\infty \leq \sqrt{\log (n/\delta)/n}$, we have 
\begin{versionC}
\[ \p_{\bPi} \sparen{ \left|~\left\| \bP_2 \bPi \by \right\|_2^2 - 1 \right| \geq \varepsilon  } \leq \delta. \]
\end{versionC}
\begin{versionA}
$ \p_{\bPi} \sparen{ \left|~\left\| \bP_2 \bPi \by \right\|_2^2 -1 \right| \geq \varepsilon  } \leq \delta.$
\end{versionA}
\end{lemma}
\begin{versionC}
{\begin{proof}
Fix $1\leq j \leq r$. Let $\by =\bW \bD \bx$.
Let $\by'=  \bPi \by$ and $\bz_{(j)} = \sqrt{n/t} \paren{\by'_{(j-1)n/r+1}, \cdots , \by'_{jn/r+1}}$. Therefore, $\|\bz_{(j)}\|^2_2 = (n/t)\sum_{i=1}^t \by_{(j-1)t+i}'^2$. From~\thmref{AC09}, we know that $\|\bz_{(j)}\|_\infty$ is bounded random variable with entries at most $\sqrt{\log(n/\delta)/t}$ with probability $1-\delta.$
Also, $\E_{\bP_2,\bPi}[\|\bz_{(j)}\|^2]=1$. Since this is sampling without replacement; therefore, we use the  Serfling bound~\cite{Serfling74}. Applying the bound gives us
\begin{align}
 \p_{\bPi} \sparen{ \left|~\left\| \bP_2 \bPi \by \right\|_2^2 - 1 \right| \geq \varepsilon  } = \p \sparen{ | \|\bz_{(j)}\|_2^2 - 1 | > \varepsilon } \leq 2\exp(-\varepsilon^2 t / \log(n/\delta)). \label{eq:P_2_1} \end{align}

In order for~\eqnref{P_2_1} to be less than $\delta/r$, we need to evaluate the value of $t$ in the following equation:
\[ \frac{\varepsilon^{-2} t}{\log (n/\delta)} \geq \log \paren{\frac{2r}{\delta}} \quad \Rightarrow \quad t \geq \varepsilon^2 \log (n/\delta) \log(2r/\delta).  \]


The inequality is trivially true for $t=n/r$ when $r = O(n^{1/2-\tau})$  for an arbitrary constant $\tau>0$. 
Now for this value of $t$, we have for a fixed $j \in [t]$
\begin{align}\p \sparen{ | \|\bz_{(j)}\|_2^2 - 1 | > \varepsilon } \leq \exp(-\varepsilon^2 t / \log(n/\delta)) \leq \delta/r. \end{align}


Using union bound over all such possible set of indices gives the lemma.
\end{proof}
\end{versionC}

In the last step, we use~\thmref{hw} to prove the following.
\begin{versionA}
A proof appears  in the full version.
\end{versionA}
\begin{lemma} \label{lem:gaussian}
	Let $\bP _1$ be a $r \times t$ random matrix with rows formed by $\bPhi, \cdots, \bPhi_r$ as above. Then for any $\bz_{(1)}, \cdots, \bz_{(r)} \in \bS^{t-1}$, we have the following:
\begin{versionC}
	\[ \p_{\mathbf{g}} \sparen{ \left| {\sum_{i=1}^r \brak{\bPhi_i, \bz_{(i)}}^2} - 1 \right| \geq \varepsilon  } \leq \delta. \]
\end{versionC}
\begin{versionA}
$ \p_{\mathbf{g}} \sparen{ \left| {\sum_{i=1}^r \brak{\bPhi_i, \bz_{(i)}}^2} - 1 \right| \geq \varepsilon  } \leq \delta. $
\end{versionA}

\end{lemma}

\begin{versionC}
\begin{proof} 
Let $\mathbf{g} $ be the vector corresponding to the matrix $\bP_1$, i.e., $\mathbf{g}_k = (\bP_1)_{ij}$ for $k=(i-1)t +j$. Let $\bA$ be the matrix formed by diagonal block matrix each of which has the form $\bz_{(i)} \bz_{(i)}^{\mathsf T}/r$ for $1\leq i \leq r$. Therefore, $\tr(\bA)=1$, $\|\bA\|_F^2 = 1/k$. Also, since the only eigen-vectors of the constructed $\bA$ is $\begin{pmatrix} \bz_{(1)} & \cdots & \bz_{(r)} \end{pmatrix}^{\mathsf T}$, $\|\bA\|=1$. Plugging these values in the result of~\thmref{hw}, we have 
\begin{align} \p_{\mathbf{g}} \sparen{ | \mathbf{g}^{\mathsf T} \bA \mathbf{g} - 1 | > \varepsilon } \leq 2 \exp \paren{ -\min \set{ \frac{c_1 \varepsilon^2}{1/r} , {c_2 \varepsilon}} } \leq \delta \label{eq:gaussian} \end{align}
for $r=O(\varepsilon^{-2} \log (1/\delta))$. 
\end{proof}

\end{versionC}

Combining~\lemref{P_2} and~\ref{lem:gaussian} and rescaling $\varepsilon$ and $\delta$ gives the following. 
\begin{versionA}
A proof appears  in the full version.
\end{versionA}
\begin{theorem}  \label{thm:newgaussian}
For any $\varepsilon, \delta >0$ and a positive integer $n$. Let $r=O(\varepsilon^{-2} \log (1/\delta)$ and  $\cD$ be a distribution of $r \times n$ matrices over $\R^{r \times n}$ 	defined as in~\figref{basic}. Then a matrix $\bPhi \sim \cD$ can be sampled using 3n random samples such that for any $\bx \in \bS^{n-1}$, 
		\begin{align*} \p_{\bPhi \sim \cD} \sparen{ \left|~\|\bPhi \bx \|_2^2 -1 \right| > \varepsilon } < \delta.  \end{align*}
		Moreover, the run-time of matrix-vector multiplication takes $O(n \log n)$ time.
\end{theorem}
\begin{versionC}
\begin{proof}
\lemref{P_2} gives us a set of $t$-dimensional vectors $\bz_{(1)}, \cdots, \bz_{(r)}$ such that each of them have Euclidean norm between $1-\varepsilon$ and $1+\varepsilon$. Using $\bz_{(1)}, \cdots, \bz_{(r)}$ to form the matrix $\bA$ as in the proof of~\lemref{gaussian},  we have $\| \bA \|_F^2 \leq (1+\varepsilon)^2/r$ and $\| \bA\| \leq (1+\varepsilon).$ Therefore, using~\thmref{hw}, we have
\begin{align} \p_{\mathbf{g}} \sparen{ | \mathbf{g}^{\mathsf T} \bA \mathbf{g} - 1 | > \varepsilon } \leq 2 \exp \paren{ -\min \set{ \frac{c_1 r \varepsilon^2}{(1+\varepsilon)^2} , \frac{c_2 \varepsilon}{(1+\varepsilon)}} } \leq \delta \label{eq:gaussian} \end{align}

For this to be less than $\delta$, we need $r \leq \frac{c'(1+\varepsilon)^2}{\varepsilon^2} \log \paren{\frac{2}{\delta}} \leq 4c' \varepsilon^{-2} \log(2/\delta)$. 
\end{proof}
\end{versionC}
The above theorem has a run-time of $O(n \log n)$ and uses $2n+n\log n$ random samples, but $O(n)$ random bits. We next show how to improve these parameters to achieve  the bounds mentioned in~\thmref{newsparse}.
\subsection{Improving the number of bits} \label{sec:random}
In the above construction, we use $2n+n\log n$ random samples. In the worse case, if $\mathbf{g}$ is picked i.i.d. from Gaussian distribution, we need to sample from a continuous distribution defined over the real\footnote{For the application in differential privacy, we are constrained to sample from $\cN(0,1)$; however, for all the other applications, we can use any subgaussian distribution.}. This requires $O(1)$ random bits for each sample and causes many numerical issues as discussed by Achlioptas~\cite{Achlioptas03}. 
We note that~\thmref{newgaussian} holds for any subgaussian random variables (the only place we need $\mathbf{g}$ is in~\lemref{gaussian}). That is, we can use $\sber(1/2)$ and our bound will still hold. This observation gives us the following result.
\begin{theorem}  \label{thm:newbernoulli}
Let $\varepsilon, \delta >0$, $n$ be a positive integer,  $\tau >0$ be an arbitrary constant, and $r=O(\varepsilon^{-2} \log (1/\delta))$ such that $r \leq n^{1/2-\tau}$, there exists a distribution of  $r \times n$ matrices over $\R^{r \times n}$ denoted by $\cD$ with the following property: a matrix $\bPhi \sim \cD$ can be sampled using only $2n+n\log n$ random bits such that for any $\bx \in \bS^{n-1}$, 
\[\p_{\bPhi \sim \cD} \sparen{ \left|~\|\bPhi \bx \|_2^2 -1 \right| > \varepsilon } < \delta.  \]
		Moreover, the run-time of matrix-vector multiplication takes $O(n \log n)$ time.
\end{theorem}

\subsection{Improving the Sparsity Factor and Run-time}  \label{sec:run}
The only dense matrix used in $\bPhi$ is the  Hadamard matrix and results in $O(n \log n)$ time. We need it so that we can apply \thmref{AC09} in the proof of~\lemref{P_2}. However, Dasgupta {\it et al.}~\cite{DKS10} showed that a sparse version of randomized  Hadamard matrix suffices for the purpose. More, specifically, the following theorem was shown by Dasgupta {\it et al.}~\cite{DKS10}.
\begin{theorem} \label{thm:DKS} 
Let $a=16 \varepsilon^{-1} \log(1/\delta) \log(r/\delta), b=6a \log(3a/\delta).$ Let $\mathbf{G} \in \R^{n\times n}$ be a random block-diagonal matrix with $n/b$ blocks of matrices of the form $ \bW_{b} \bB$ matrix, where $\bB$ is a $b \times b$ diagonal matrix with non-zero entries picked i.i.d. from $\sber(1/2)$. Then we have
$ \p_{\mathbf{G}} \sparen{ \|\mathbf{G} \bx\|_\infty \geq a^{-1/2} } \leq \delta. $
\end{theorem}

We can apply~\thmref{DKS} instead of~\thmref{AC09} in our proof of~\thmref{newbernoulli} as long as $1/a \leq \log (n/\delta)/n$. This in turn implies that we can use this lemma in our proof as long as 
$ \varepsilon \leq \paren{ \log(1/\delta) \log (n/\delta) \log^2(r/\delta)}/n. $
When this condition is fulfilled, $\mathbf{G}$ takes $n/b (b \log b)=n \log b$ time to compute matrix-vector multiplication. Since $\bP$ and $\bPi$ takes linear time, using $\mathbf{G}$ instead of randomized  Hadamard matrices and noting that $n/r$ is still permissible, we get~\thmref{newsparse}. 
If $\varepsilon n \geq \log(1/\delta) \log (n/\delta) \log^2(r/\delta) $, we have to argue further. A crucial point to note is that~\thmref{DKS} is true for any $a>1$. Using this observation, we show that the run time is $O(n \log r)$ in this case.  For the sake of completion, we present more details  in the full version\footnote{Alternatively, we can also use the idea of Ailon and Liberty~\cite[Sec 5.3]{AL09} and arrive at the same conclusion.}. 

\begin{versionC}
For a sparse vector $\bx$, we can get even better run-time.   Let $\mathsf{s}_\bx$ be the number of non-zero entries in a vector $\bx$. Then the running time of the block randomized Hadamard transform of~\thmref{DKS} is $O(\mathsf{s}_\bx b \log b + \mathsf{s}_\bx b). $ Plugging in the value of $b$, we get the final running time to be

 \[ O \paren{ \min \set{ \frac{\mathsf{s}_\bx}{\varepsilon} \log \paren{\frac{1}{\delta}} \log^2 \paren{\frac{k}{\delta}} \log \paren{\frac{1}{\delta \varepsilon}}, n } \log \paren{\frac{1}{\delta \varepsilon}} }. \]
\end{versionC}

\section{Applications in Differentially Private Algorithms} \label{sec:applications} \label{sec:DP}
In this section, we discuss the applications of our construction in differential privacy. Since $\bW$, $\bD$, and $\bPi$ are isometry, it seems that we can use the proof of~\cite{BBDS12,Upadhyay13} to prove that $\bP \bA^{\mathsf T}$ preserves differential privacy when $\bP$ is picked using Gaussian variables $\cN(0,1)$ and $\bA$ is an $m \times n$ private matrix. Unfortunately, the following attack shows this is not true.
\begin{versionC}
Let $\bP$ be as in~\figref{basic} and two neighbouring matrices $\bA$ and $\widetilde{\bA}$ be as follows:
\[ \bA= 	\begin{pmatrix} w  & \mathbf{0}^{n/r-1}  &
  				\mathbf{0}^{n/r}  &
  				\cdots  &
  				\mathbf{0}^{n/r} 
			\end{pmatrix} \quad \text{and} \quad	
\widetilde{\bA}= 	\begin{pmatrix} w  & \mathbf{0}^{n/r-1}  &
  				\mathbf{1}^{n/r}/\sqrt{n/r} &
  				\cdots  &
  				\mathbf{0}^{n/r} 
			\end{pmatrix},				
 \]  
 where $\mathbf{1}^{n/r}$ is all one row vector of dimension $n/r$. It is easy to see that $\bP \bA$ has zero entries at positions $n/r+1$ to $2n/r$, while $\bP \widetilde{\bA}$  has non-zero entries. Thereby, it breaches the privacy.
\end{versionC} 
\begin{versionA}
Let the  neighboring matrices be 
$\bA= 	\begin{pmatrix} w &\mathbf{0}^{n/r-1}  &
  				\mathbf{0}^{n/r}  &
  				\cdots  &
  				\mathbf{0}^{n/r} 
			\end{pmatrix}~\text{and}~	
\widetilde{\bA}= 	\begin{pmatrix} w & \mathbf{0}^{n/r-1}  &
  				\mathbf{1}^{n/r}/\sqrt{n/r} &
  				\cdots  &
  				\mathbf{0}^{n/r} 
			\end{pmatrix}	$, 
			 where $\mathbf{1}^{n/r}$ is all one row vector of dimension $n/r$. It is easy to see that $\bP \bA$ has zero entries  while $\bP \widetilde{\bA}$  has non-zero entries at positions $n/r+1$ to $2n/r$. This violates differential privacy.
\end{versionA}
This attack shows that we have to analyze the action of $\bPhi$ on the private matrix $\bA$. 

We assume that $m \leq n$. Let $\by=\bPi  \bW \bD\bx$ for a vector $\bx$.  $\bPhi$ is a linear map; therefore, without loss of generality, we can assume that $\bx$ is a unit vector. Let $\by^{(1)}, \cdots, \by^{(r)}$ be disjoint blocks of the vector $\by$ such that each blocks have same number of entries. The first observation is that even if we have single non-zero entries in $\by^{(i)}$ for all $1 \leq i \leq r$, then we can extend the proof of~\cite{BBDS12, Upadhyay14} to prove differential privacy. In~\lemref{zero}, we show that with high probability, $\by^{(i)}$ is not all zero vector for $1 \leq i \leq r$. Since Euclidean norm is a metric, we prove this by proving that the Euclidean norm of each of the vectors $\by^{(i)}$ is non-zero. Formally,
\begin{lemma} \label{lem:zero}
Let $\bx \in \bS^{n-1}$ be an arbitrary unit vector and $\by=\bPi  \bW \bD\bx$. Let $\by^{(1)}, \cdots, \by^{(r)}$ be the $r$ disjoint blocks of the vector $\by$ such that $\by^{(i)} =(\by_{(i-1)t+1}, \cdots, \by_{it})$ for $t=n/r$. 
 Then for an arbitrary $0< \theta <1$  and for all $i \in [r]$, 
\begin{versionA}
$ \p_{\bPi,\bD} \sparen{ \sqrt{n/t} \left\| \by^{(i)}  \right\|_2  \leq \theta  } \leq 2^{-\Omega((1 - \theta)^2 n^{2/3})}. $
\end{versionA}
\begin{versionC}
\[ \p_{\bPi,\bD} \sparen{  \left\| \by^{(i)}  \right\|_2  \leq \theta  } \leq 2^{-\Omega((1 - \theta)^2 n^{2/3})}. \]
\end{versionC}
\end{lemma}

\begin{versionC}
{\begin{proof}
Let us denote by $t=n/r$. From~\thmref{AC09}, we know that $\bW \bD \bx$ is bounded random variable with entries at most $\sqrt{(\log(n)+n^{1/3})/n}$ with probability $1-2^{-c_3n/3}.$   Let $\by= \bPi \bW \bD \bx$, and  $\by^{(j)} = \sqrt{n/t}(\by_{(j-1)t+1}, \cdots, \by_{jt})$

Fix a $j$. First note that $\E[\|\by^{(j)}\|^2]=1$. Since the action of $\bPi$ results in sampling without replacement; therefore, Serfling's bound~\cite{Serfling74} gives
\begin{align*}
\p \sparen{ \left\| \by^{(j)} \right\|_2^2  \leq 1- \zeta } \leq \exp \paren{-\frac{\zeta^2 n }{ \log(n)+n^{1/3}}} \leq \exp(- c_3' \zeta^2 n^{2/3})
 \end{align*}
 for $\zeta <1$.
Setting $\theta = 1-\zeta$ and using union bound over all $j$ gives the lemma. 
\end{proof}
}
\end{versionC}
\begin{versionA}
We prove this lemma in the full version.
\end{versionA}
Therefore, the above attack cannot be carried out in this case. However, this does not show that publishing $\bPhi \bA^{\mathsf T}$ preserves differential privacy. To prove this, we need to argue further. 
We use~\lemref{zero} to prove the following result.
\begin{lemma} \label{lem:singular}
	Let $t$ be as in~\lemref{zero}. Let $\bPi_{1..t}$ denotes a permutation matrix on $\set{1,\cdots, n}$ restricted to any consecutive $t$ rows. Let $\bW$ be the $n \times n$ Walsh-Hadamard matrix and $\bD$ be a diagonal Rademacher matrix. Let $\mathbf{B}=\sqrt{N/t}\bPhi \bW \bD$. Then $$ (1-\varepsilon)\I \preceq \mathbf{B}^{\mathsf T} \mathbf{B} \preceq (1+\varepsilon) \I.$$ 
\end{lemma}
In other words, what~\lemref{singular} says is that the singular values of $\mathbf{B}$ is between $(1\pm \varepsilon)^{1/2}$. This is because all the singular values of $\I$ are $1$.

\begin{proof}
We have the following:
\begin{align*}
	\| \mathbf{B}^{\mathsf T} \mathbf{B}  \|_2 &=  \max_{\bx \in \R^n} \frac{\bx^{\mathsf T} \mathbf{B}^{\mathsf T} \mathbf{B} \bx}{\brak{\bx,\bx}}  \\
							&= \max_{\bx \in \R^n} \frac{\brak{\mathbf{B} \bx, \mathbf{B} \bx}}{\brak{\bx,\bx}} \\
							& = \max_{\| \bx\|_2 =1} \| \mathbf{B} \bx\|_2^2 . 
\end{align*}
From~\lemref{zero}, with probability at least $1 -\delta$, for all $\bx$,
\begin{align} 1-\varepsilon \leq \max_{\| \bx\|_2 =1} \| \mathbf{B} \bx\|_2^2 \leq 1 + \varepsilon. \label{eq:singular} \end{align}
The lemma follows from the observation that~\eqnref{singular} holds for every $\bx \in \R^n$. In particular, if $\set{\mathbf{v}_1, \cdots, \mathbf{v}_k}$ are the first $k$ eigen-vectors of $\mathbf{B}^{\mathsf T} \mathbf{B}$, then it holds for the subspace of $\R^n$ orthogonal to the subspace spanned by the vectors $\set{\mathbf{v}_1, \cdots, \mathbf{v}_k}$.
\end{proof}
Once we have this lemma, the following is a simple corollary.
\begin{corollary} \label{cor:singular}
Let $\bA$ be an $n \times d$ matrix such that all the singular values of $\bA$ are at least $\sigma_\mathsf{min}$. Let the matrix $\mathbf{B}$ be as formed in the statement of~\lemref{singular}. Then all the singular values of $\mathbf{B} \bA$ are at least $\sqrt{1-\varepsilon} \sigma_\mathsf{min}$ with probability $1-\delta$.
\end{corollary}

Combining~\corref{singular} with the result of Upadhyay~\cite{Upadhyay14} yields the following result. 

\begin{theorem} \label{thm:DP}
	If the singular values of an $m \times n$ input matrix $\bA$ are at least $\frac{\ln(4/\beta)\sqrt{16 r \log (2/\beta)}}{\alpha(1+\varepsilon)}$, then $\bPhi \bA^{\mathsf T}$, where $\bPhi$ is as in~\figref{basic} is $(\alpha, \beta+\delta)$-differentially private. It uses $2n+n\log n$ samples and takes $O(mn \log b)$ time, where $ b= \min \{ r,c_1 \varepsilon^{-1} \log(1/\delta) \log(r/\delta) \log((16 \varepsilon^{-1} \log(1/\delta) \log(r/\delta))/\delta)\}$. 
\end{theorem}
The extra $\delta$ term in~\thmref{DP} is due to~\corref{singular}. 
\begin{versionA}
We give a detail proof in the full version.
\end{versionA}
\begin{versionC}
\begin{proof}
Let $\bA$ be the private input matrix and $\widetilde{\bA}$ be the neighbouring matrix. In other word, there exists a unit vector $\mathbf{v}$ such that $\bA - \widetilde{\bA}= \bE= \mathbf{v} \mathbf{e}_i^{\mathsf T}  .$ Let $\bU \bSigma \bV^{\mathsf{T}}$ ($\widetilde{\bU} \widetilde{\bSigma} \widetilde{\bV}^{\mathsf{T}}$, respectively) be the singular value decomposition of $\bA$ ($\widetilde{\bA}$, respectively). Further, since the singular values of $\bA$ and $\widetilde{\bA}$ is at least $\sigma_{\mathsf{min}}=\paren{ \frac{ \sqrt{r\log (2/\beta)} \log (r/ \beta)}{\alpha (1+\varepsilon)}}$, we can write 
\[ \bA = \bU (\sqrt{\bLambda^2 + \sigma_{\mathsf{min}}^2 \I }) \bV^{\mathsf T}~\text{ and}~
  \widetilde{\bA} = \widetilde{\bU} (\sqrt{\widetilde{\bLambda}^2 + \sigma_{\mathsf{min}}^2 \I }) \widetilde{\bV}^{\mathsf T}
 \]
for some  diagonal matrices $\bLambda$ and $\widetilde{\bLambda}$.

The basic idea of the proof is as follows. Recall that we analyzed the action of $\bPhi$ on a unit vector $\bx$ through a series of composition of matrices. In that composition, the last matrix was $\bP_1$, a random Gaussian matrix. We know that this matrix preserves differential privacy if the input to it has certain spectral property. Therefore, we reduce our proof to proving this; the result then follows using the proof of~\cite{BBDS12,Upadhyay14}.

Now consider any row $j$ of the published matrix. It has the form 
\begin{align} (\bPhi \bA^{\mathsf T})_{j:} = \paren{\bPhi \bV  (\sqrt{\bLambda^2 + \sigma_{\mathsf{min}}^2 \I })  \bU^{\mathsf T}}_{j:}~\text{and}~(\bPhi \widetilde{\bA}^{\mathsf T})_{j:} = \paren{\bPhi \widetilde{\bV}  (\sqrt{\widetilde{\bLambda}^2 + \sigma_{\mathsf{min}}^2 \I })  \widetilde{\bU}^{\mathsf T}}_{j:}, \label{eq:rowi}
\end{align} respectively. In what follows, we show that the distribution for the first row of the published matrix is differentially private; the proof for the rest of the rows follows the same proof technique.~\thmref{DP} is then obtained using~\thmref{DRV10}. Expanding on the terms in~\eqnref{rowi} with $i=1$, we know that the output is distributed as per $\bPhi_1 (\bPi \bW \bD \bV)_{1..t} (\sqrt{\bLambda^2 + \sigma_{\mathsf{min}}^2 \I })  \bU^{\mathsf T}$, where $(\bPi \bW \bD \bV)_{1..t}$ represents the first $t=n/r$ rows of the matrix $\bPi \bW \bD \bV$. Using the notation of~\lemref{singular}, we write it as $\bPhi_1 \mathbf{B} \bV (\sqrt{\bLambda^2 + \sigma_{\mathsf{min}}^2 \I })  \bU^{\mathsf T}$
 From~\corref{singular}, we know that with probability at least $1-\delta$, we know that the singular values of $\mathbf{B} \mathbf{A}^{\mathsf T}$ are at least $\sqrt{1-\varepsilon} \sigma_\mathsf{min}$. With a slight abuse of notation, let us denote by $ \bV' (\sqrt{\bLambda^2 + \sigma_{\mathsf{min}}^2 \I })  \bU^{\mathsf T}$ the singular value decomposition of $\mathbf{B} \bA$. We perform the same substitution to derive the singular value decomposition of $\mathbf{B}\widetilde{\bA}^{\mathsf T}$. Let us denote it by $ \widetilde{\bV}' (\sqrt{\widetilde{\bLambda}^2 + \sigma_{\mathsf{min}}^2 \I })  \widetilde{\bU}^{\mathsf T}$.  

Since $\bPhi_1 \sim \cN(\mathbf{0}^t, \I_{t \times t})$, the probability density function corresponding to $\bPhi_1 \bV' (\sqrt{\bLambda^2 + \sigma_{\mathsf{min}}^2 \I })  \bU^{\mathsf T}$ and $\bPhi_1 \widetilde{\bV}' (\sqrt{\widetilde{\bLambda}^2 + \sigma_{\mathsf{min}}^2 \I })  \widetilde{\bU}^{\mathsf T}$ is as follows:
\begin{align}
\frac{1}{ \sqrt{(2\pi)^d \det((\bV' \bSigma \bU^{\mathsf T}) \bU \bSigma \bV'^{\mathsf T}) } } \exp \paren{- \frac{1}{2} \bx^{\mathsf T} \paren{ (\bV' \bSigma \bU^{\mathsf T}) \bU \bSigma \bV'^{\mathsf T} }^{-1} {\bx}}  \\
\frac{1}{\sqrt{(2\pi)^d \det((\widetilde{\bV}' \widetilde{\bSigma} \widetilde{\bU}^{\mathsf T}) \widetilde{\bU} \widetilde{\bSigma} \widetilde{\bV}'^{\mathsf T})}} \exp  \paren{- \frac{1}{2} \bx^{\mathsf T} \paren{(\widetilde{\bV}' \widetilde{\bSigma} \widetilde{\bU}^{\mathsf T}) \widetilde{\bU} \widetilde{\bSigma} \widetilde{\bV}'^{\mathsf T}}^{-1} {\bx} },
\end{align}
where $\det(\cdot)$ represents the determinant of the matrix.

We first prove the following:
\begin{align} 
\exp\paren{- \frac{\alpha}{\sqrt{4 r \ln (2/\beta)}}} \leq  \sqrt{\frac {\det((\widetilde{\bV}' \widetilde{\bSigma} \widetilde{\bU}^{\mathsf T}) \widetilde{\bU} \widetilde{\bSigma} \widetilde{\bV}'^{\mathsf T}) } {\det((\bV' \bSigma \bU^{\mathsf T}) \bU \bSigma \bV'^{\mathsf T})}} \leq  \exp\paren{ \frac{\alpha}{\sqrt{4 r \ln (2/\beta)}}}  \label{eq:first}.
\end{align} 

This part follows simply as in Blocki {\it et al.}~\cite{BBDS12}. More concretely, we have $\det((\bV' \bSigma \bU^{\mathsf T}) \bU \bSigma \bV'^{\mathsf T}) = \prod_i \sigma_i^2$, where $\sigma_1 \geq \cdots \geq \sigma_m \geq  \sigma_{\mathsf{min}}$ are the singular values of $\bA$. 
 Let $\widetilde{\sigma}_1 \geq \cdots \geq \widetilde{\sigma}_m \geq  \sigma_{\mathsf{min}}$ be its singular value  for $\widetilde{\bA}$. Since the singular values of $\bA - \widetilde{\bA}$ and $\widetilde{\bA} -\bA$ are the same,  $\sum_i(\sigma_i - \widetilde{\sigma}_i) \leq 1$ using the Linskii's theorem. Therefore, 
\begin{align*} 
\sqrt{\frac{\det((\widetilde{\bV}' \widetilde{\bSigma}^2 \widetilde{\bV}'^{\mathsf T}) } {\det((\bV' \bSigma^2 \bV'^{\mathsf T})}}  = \sqrt{\prod_i \frac{\widetilde{\sigma}_i^2}{\sigma_i^2}} \leq \exp \paren{\frac{\alpha}{32 \sqrt{r \log (2/\beta)} \log (r/\beta)} }\sum_i (\widetilde{\sigma}_i - \sigma_i) \leq e^{\alpha_0/2}. \end{align*}

Let $\beta_0 =\beta/2r$. In the second stage, we prove the following:
\begin{align}
 \p \sparen{   \left|
 \bx^{\mathsf T} \paren{ (\bV' \bSigma \bU^{\mathsf T}) \bU \bSigma \bV'^{\mathsf T} }^{-1} {\bx} 
 - 
 \bx^{\mathsf T} \paren{(\widetilde{\bV}' \widetilde{\bSigma} \widetilde{\bU}^{\mathsf T}) \widetilde{\bU} \widetilde{\bSigma} \widetilde{\bV}'^{\mathsf T}}^{-1} {\bx} \right| \leq \alpha_0} \geq 1 -\beta_0.
  \label{eq:second} 
\end{align}
Every row of the published matrix is distributed identically; therefore, it suffices to analyze the first row. The first row is constructed by multiplying  a $t$-dimensional  vector $\bPhi_1$ that has entries 
picked 
from a normal distribution $\cN(0,1)$. Note that $\E[\bPhi_{1}]=\mathbf{0}^t$ and $\cov(\bPhi_i)=\I$. Therefore, using the fact that $\bU \bSigma \bV'^{\mathsf T} - \widetilde{\bU} \widetilde{\bSigma} \widetilde{\bV}'^{\mathsf T} = \mathbf{E}' = \mathbf{v}' \mathbf{e}_i^{\mathsf T}$, where $\mathbf{v}'$ is a vector restricted to the first $t$ columns of $\mathbf{v}$.
\begin{align*} 
& \bx^{\mathsf T} \paren{ (\bV' \bSigma \bU^{\mathsf T}) \bU \bSigma \bV'^{\mathsf T} }^{-1} {\bx}  - 
 \bx^{\mathsf T} \paren{(\widetilde{\bV}' \widetilde{\bSigma} \widetilde{\bU}^{\mathsf T}) \widetilde{\bU} \widetilde{\bSigma} \widetilde{\bV}'^{\mathsf T}}^{-1} {\bx} \\
 & \qquad = \bx^{\mathsf T} \sparen{ \paren{ (\bV' \bSigma \bU^{\mathsf T}) \bU \bSigma \bV'^{\mathsf T} }^{-1}   -  \paren{(\widetilde{\bV}' \widetilde{\bSigma} \widetilde{\bU}^{\mathsf T}) \widetilde{\bU} \widetilde{\bSigma} \widetilde{\bV}'^{\mathsf T}}^{-1}} {\bx}  \\
 & \qquad = \bx^{\mathsf T} \sparen{ \paren{ (\bV' \bSigma \bU^{\mathsf T}) \bU \bSigma \bV'^{\mathsf T} }^{-1} \paren{ (\bV' \bSigma \bU^{\mathsf T}) \mathbf{E}' + \mathbf{E}'^{\mathsf T} \widetilde{\bU} \widetilde{\bSigma} \widetilde{\bV}'^{\mathsf T}  }  \paren{(\widetilde{\bV}' \widetilde{\bSigma} \widetilde{\bU}^{\mathsf T}) \widetilde{\bU} \widetilde{\bSigma} \widetilde{\bV}'^{\mathsf T}}^{-1}} {\bx} .
\intertext{Also $\bx^{\mathsf T} = \bPhi_1 \bV' \bSigma  \bU^{\mathsf T} $. This further simplifies to}
 & \qquad = \bPhi_1 \bV' \bSigma  \bU^{\mathsf T} \sparen{ \paren{ \bV' \bSigma^2 \bV'^{\mathsf T} }^{-1} \paren{ (\bV' \bSigma \bU^{\mathsf T}) \mathbf{E}' + \mathbf{E}'^{\mathsf T} \widetilde{\bU} \widetilde{\bSigma} \widetilde{\bV}'^{\mathsf T}  }  \paren{(\widetilde{\bV}' \widetilde{\bSigma}^2 \widetilde{\bV}'^{\mathsf T}}^{-1}}  \bU \bSigma  \bV'^{\mathsf T} \bPhi_1^{\mathsf T} \\
 & \qquad = \bPhi_1 \sparen{ \paren{\bV' \bSigma  \bU^{\mathsf T} } \paren{ \bV' \bSigma^2 \bV'^{\mathsf T} }^{-1} (\bV' \bSigma \bU^{\mathsf T}) \mathbf{E}'  \paren{\widetilde{\bV}' \widetilde{\bSigma}^2 \widetilde{\bV}'^{\mathsf T}}^{-1}  \bU \bSigma  \bV'^{\mathsf T}} \bPhi_1^{\mathsf T}   \\
 & \qquad \quad +\bPhi_1 \sparen{ \paren{ \bV' \bSigma  \bU^{\mathsf T} } \paren{ \bV' \bSigma^2 \bV'^{\mathsf T} }^{-1}  \mathbf{E}'^{\mathsf T} \paren{\widetilde{\bU} \widetilde{\bSigma} \widetilde{\bV}'^{\mathsf T} }   \paren{\widetilde{\bV}' \widetilde{\bSigma}^2 \widetilde{\bV}'^{\mathsf T}}^{-1}  \bU \bSigma  \bV'^{\mathsf T}} \bPhi_1^{\mathsf T} .
 \end{align*}

Using the fact that $\mathbf{E}'= \mathbf{v}'\mathbf{e}_i^{\mathsf T}$ for some $i$, we can write the above expression in the form of $t_1 t_2 + t_3t_4$, where 
\begin{align*}
t_1&=\bPhi_1 \paren{\bV' \bSigma  \bU^{\mathsf T} } \paren{ \bV' \bSigma^2 \bV'^{\mathsf T} }^{-1} (\bV' \bSigma \bU^{\mathsf T}) \mathbf{v}', \\
t_2 &= \mathbf{e}_i^{\mathsf T}  \paren{\widetilde{\bV}' \widetilde{\bSigma}^2 \widetilde{\bV}'^{\mathsf T}}^{-1}  \bU \bSigma  \bV'^{\mathsf T} \bPhi_1^{\mathsf T}, \\
t_3 &= \bPhi_1 \paren{ \bV' \bSigma  \bU^{\mathsf T} } \paren{ \bV' \bSigma^2 \bV'^{\mathsf T} }^{-1}  \mathbf{e}_i, \\
t_4 &=\mathbf{v}'^{\mathsf T} \paren{\widetilde{\bU} \widetilde{\bSigma} \widetilde{\bV}'^{\mathsf T} }   \paren{\widetilde{\bV}' \widetilde{\bSigma}^2 \widetilde{\bV}'^{\mathsf T}}^{-1}  \bU \bSigma  \bV'^{\mathsf T} \bPhi_1^{\mathsf T}.
\end{align*}


Recall that $ \sigma_{\mathsf{min}}=\paren{ \frac{ \sqrt{r\log (2/\beta)} \log (r/ \beta)}{\alpha}}$, $\bU \bSigma \bV'^{\mathsf T} - \widetilde{\bU} \widetilde{\bSigma} \widetilde{\bV}'^{\mathsf T}  = \mathbf{v}' \mathbf{e}_i^{\mathsf T} $.
 Now since $\| \widetilde{\bSigma} \|_2, \| \bSigma \|_2 \geq  \sigma_{\mathsf{min}}$, $\| \mathbf{v}' \| \leq \| \mathbf{v} \| \leq 1$, and that every term $t_i$ in the above expression is a linear combination of a Gaussian, i.e., each term is distributed as per $ \cN(0,\|t_i\|^2)$, we have 
\begin{align*}
\|t_1\| \leq 1,~\|t_2 \| \leq  \frac{ 1}{ \sigma_{\mathsf{min}}},~\|t_3\| \leq \frac{1}{ \sigma_{\mathsf{min}}} + \frac{1 }{ \sigma_{\mathsf{min}}^2},~\|t_4\| \leq 1  + \frac{1 }{ \sigma_{\mathsf{min}}}.
\end{align*}


 Using the concentration bound on the Gaussian distribution, each term, $t_1,t_2,t_3$, and $t_4$, is less than $\|t_i\| \ln (4/\beta_0)$ with probability $1 - \beta_0/2$. From the fact that $2\paren{\frac{1}{ \sigma_{\mathsf{min}}} + \frac{1}{ \sigma_{\mathsf{min}}^2}} \ln (4/\beta_0) \leq \alpha_0$, we have~\eqnref{second}.
 
 Combining~\eqnref{first} and~\eqnref{second} with~\lemref{zero} gives us~\thmref{DP}.

\end{proof}
\end{versionC}
\thmref{DP} improves the efficiency of both the algorithms of~\cite{BBDS12} (without sparsification trick~\cite{Upadhyay13}). 
\begin{theorem}
One can $(\alpha,\beta+2^{-\Omega((1 - \theta)^2 n^{2/3})})$-differentially private  answer both cut-queries on a graph and directional covariance queries of a matrix in time $O(n^2 \log r)$. Let $\cQ$ be the set of queries  such that $|\cQ| \leq 2^r$. The total additive error incurred to answer all $\cQ$ queries is $O(|S| \sqrt{\log |\cQ|/\alpha})$ for cut-queries on a graph, where $S$ is the number of vertices in the cut-set, and $O(\alpha^{-2} \log |\cQ|)$ for directional covariance queries.
\end{theorem}

\begin{versionA}
\paragraph{Differentially-Private Manifold Learning.}
\end{versionA}
\begin{versionC}
\subsection{Differentially private manifold learning}
 We next show how we can use~\thmref{DP} to differentially private manifold learning.
\end{versionC}
The manifold learning is often considered as non-linear analog of principal component analysis. It is defined as follows: given $\bx_1,\cdots, \bx_m \in \R^n$ that lie on an $k$-dimensional manifold $\mathscr{M}$ that can be described by $f:\mathscr{M} \rightarrow \R^n$, find $\by_1, \cdots, \by_m$ such that $\by_i=f(\bx_i)$. We consider two set of sample points neighbouring if they differ by exactly one points which are apart by Euclidean distance at most $1$.~\thmref{DP} allows us to transform a known algorithm ({\em Isomap}~\cite{TDL00}) for manifold learning to a non-interactive differentially private algorithm under a smoothness condition on the manifold. 
The design of our algorithm require some care to restrict the additive error.
We guarantee that the requirements of Isomap is fulfilled by our algorithm. 
\begin{versionA}
We defer the algorithm and our proof of the following theorem in the full version.
\end{versionA}
\begin{theorem} \label{thm:manifold}
Let $\mathscr{M}$ be a compact $k$-dimensional Riemannian manifold of $\R^n$ having condition number $\kappa$, volume $\mathsf{V}$ , and geodesic covering regularity $\mathsf{r}$. If $n^{1/2-\tau} \varepsilon^2 \geq c k \log (n \mathsf{V}/\kappa) \log(1/\delta))$ for an universal constant $c$, then we can learn the manifold in an $(\varepsilon,\delta+2^{-\Omega((1 - \theta)^2 n^{2/3})} )$-differentially private manner.
\end{theorem}
\begin{versionC}
\begin{proof} We  convert the non-private algorithm of Hegde, Wakin, and Baraniuk~\cite{HWB} to a private algorithm for manifold learning. We consider two set of sample points neighbouring if they differ by exactly one point and the Euclidean distance between those two different points is  at most $1$. An important parameter for all manifold learning algorithms is the {\em intrinsic dimension} of a point. Hegde {\it et al.}~\cite{HWB} focussed their attention on the Grassberger and Procaccia's algorithm for the estimation of intrinsic dimension~\cite{GP}. It computes the {\em correlated dimension} of the sampled point. 
\begin{definition}
	Let $X=\set{\bx_1, \cdots, \bx_m}$ be a finite data points of the underlying dimension $k$. Then define 
	\[ C_m(d) = \frac{1}{m(m-1)} \sum_{i \neq j} \mathbf{1}[\| \bx_i - \bx_i \|_2 <r], \]
	where $\mathbf{1}$ is the indicator function. Then the {\em correlated dimension} of $\mathbf{X}$ is defined as 
	\[ \widehat{k} (d_1, d_2) := \frac{\log (C_m(d_1)/C_m(d_2))}{\log (d_1/d_2)}. \]
\end{definition}

We embed the data points to a low-dimension subspace. It is important to note that if the embedding dimension is very small, then it is highly likely that distinct point wold be mapped to same embedded point. Therefore, we need a find the minimal embedding subspace that suffices for our purpose. 

Our analysis would centre around a very popular algorithm for manifold learning known as {\em Isomap}~\cite{TDL00}. It is a non-linear algorithm that is used to generate a mapping from a set of sampled points in a $k$-dimensional manifold into a Euclidean space of dimension $k$.  It attempts to preserve the metric structure of the manifold, i.e., the geodesic distances between the points. The idea is to define a suitable graph which approximates the geodesic distances and then perform classical multi-dimensional scaling to obtain a reduced $k$-dimensional representation of the data. 

There are two key parameters in the Isomap algorithm that we need to control: intrinsic dimension and the {\em residual variance} which measures how well a given dataset can be embedded not a $k$-dimensional Euclidean space. 

Baraniuk and Wakin~\cite{BW09} proved the following result on the distortion suffered by pairwise geodesic distances and Euclidean distances if we are given certain number of sample points.
\begin{theorem} \label{thm:geodesic}
Let $\mathscr{M}$ be a compact $k$-dimensional manifold in $\R^n$ having volume $\mathsf{V}$ and condition number $\kappa$. Let $\bPhi$ be a random projection that satisfies~\thmref{JL} such that 
\[ r \geq O \paren{\frac{k \log (n \mathsf{V} \kappa) \log(1/\delta)}{ \varepsilon^{2}} }.  \]
Then with probability $1-\delta$, for every pair of points $\bx , \by \in \mathscr{M}$, 
\[ (1-\varepsilon) \sqrt{r/n} \| \bx - \by \|_{\mathsf{Geo}} \leq \| \bPhi \bx - \bPhi \by \|_{\mathsf{Geo}} \leq (1+\varepsilon) \sqrt{r/n} \| \bx - \by \|_{\mathsf{Geo}}, \]
where $\| \bx - \by \|_{\mathsf{Geo}}$ stands for the geodesic distance between $\bx$ and $\by$.
\end{theorem}

Therefore, we know that the embedded points represents the original manifold very well. To round things up, we need to estimate how well Isomap performs when only the embedded vectors are given. 
The Isomap algorithm for manifold learning uses a preprocessing step that estimates the intrinsic dimension of the manifold. The algorithm mentioned in~\figref{basic} computes an intrinsic dimension $\widehat{k}$ and a minimal dimension of the embedding $r$ before it invokes the Isomap algorithm on the embedded points. The privacy is preserved because of our transformation of sampled points $\set{\by_i}$ to a set of points $\set{\bx_i}$. This is because the matrix $\mathbf{X}$ formed by the set of points $\set{\bx_i}$ has its singular values greater than $\sigma_{\mathsf{min}}$ required in~\thmref{DP}. So in order to show that the algorithm mentioned in~\figref{manifold} learns the manifold, we need to show that the intrinsic dimension is well estimated and there is a bound on the residual variance due to embedding.
In other words, we need results for the following two and we are done.
\begin{enumerate}
	\item Preprocessing Stage: We need an  bound on {\em correlated dimension} of the sample points in the private setting.  
	\item Algorithmic Stage: We need to bound the error on the residual variance in the private setting.
\end{enumerate}

Once we have the two results, the algorithm for manifold learning follows using the standard methods of Isomap~\cite{TDL00}. For the correctness of the algorithm, refer to Hegde {\it et al.}~\cite{HWB} and the original work on Isomap~\cite{TDL00}. Here we just bound the correlated dimension and the residual variance. For the former, we use the algorithm of Grassberger and Procaccia~\cite{GP}. For the latter, we analyze the effect of perturbing the input points so as to guarantee privacy. 

\begin{figure} [t]
\begin{center}
\fbox{
\begin{minipage}[l]{6in}
\medskip
{\bf Input:} A set of sampled point $Y=\set{\by_1, \cdots, \by_m}$ in $\R^n$. \\
{\bf Initialization:} $r :=1$.  Construct $\bx_i \in \R^{2n}$ by setting $\bx_i = \begin{pmatrix} \sigma_{\mathsf{min}} \mathbf{e}_i & \by_i \end{pmatrix}^{\mathsf T}$ for all $1 \leq i \leq m$. Construct the matrix $\mathbf{X}$ with columns $\set{\bx_1, \cdots, \bx_m}$.\\

while residual variance $\geq \zeta$ do
\begin{itemize}
	\item Sample $\bPhi$ from a distribution over $r \times 2n$ random matrices defined in~\figref{basic}
	\item Run the Grassberger-Procaccia's algorithm on $\bPhi \mathbf{X}$.
	\item Use the estimate $\widehat{k}$ to perform Isomap on $\bPhi \mathbf{X}$.
	\item Calculate the residual variance.
	\item Increment $r$.
\end{itemize}
	Invoke the Isomap algorithm~\cite{TDL00} with the intrinsic dimension $\widehat{k}$, points $\bx_1, \cdots, \bx_m$ as sample points, and an $r \times n$ matrix picked as per~\figref{basic}.
\end{minipage}
}\caption{Differentially Private Manifold Learning} \label{fig:manifold} 
\end{center}
\end{figure}

\begin{remark}
Before we proceed with the proofs, a remark is due here. Note that we need to compute $\bPhi \mathbf{X}$ until the residual error is more than some threshold of tolerance. If we simply use Gaussian random matrices this would require $O(rn^2)$ time per iteration, an order comparable to the run-time of Isomap~\cite{TDL00}. On other hand, using our $\bPhi$, we need $O(r n \log r)$ time per iteration. This leads to a considerable improvement in the run-time of the learning algorithm.
\end{remark}
Since the proof of Hegde {\it et al.}~\cite{HWB} for the preprocessing step does not need the specifics of the sampled point, but only use the properties of $\bPhi$ in preserving the pairwise Euclidean distance, the first bound follows as in Hegde {\it et al.}~\cite{HWB}. 
\begin{theorem}
Let $\mathscr{M}$ be a compact $k$-dimensional manifold in $\R^n$ having volume $\mathsf{V}$ and condition number $\kappa$. Let $\set{\bx_1, \cdots, \bx_m}$ be a finite data points of the underlying dimension $k$.  Let $\widehat{k}$ be the dimension estimate of the algorithm of Grassberger and Procaccia~\cite{GP} over the range $(d_{\mathsf{min}}, d_{\mathsf{max}}).$ Let $\eta=-\log (d_{\mathsf{min}}/ d_{\mathsf{max}})$. Suppose  $d_{\mathsf{max}} \leq \kappa/2$. Fix a $0<\eta <1$ and $0 < \delta <1$.
Let $\bPhi$ be a random projection that satisfies~\thmref{JL} such that 
\[ r \geq O \paren{\frac{k \log (n \mathsf{V} \kappa) \log(1/\delta)}{ \eta^2 \zeta^{2}} }.  \]
Let $\widehat{k}_{\bPhi}$ be the estimated correlation dimension on $\bPhi X$ over the range $(d_{\mathsf{min}}\sqrt{r/n}, d_{\mathsf{max}}\sqrt{r/n}).$ Then with probability at least $1-\delta$,
\[ (1-\eta)\widehat{k} \leq \widehat{k}_{\bPhi} \leq (1+\eta)\widehat{k}.  \]
\end{theorem}

For the second part, we need  to take care of two types of perturbation: the input perturbation to convert the original sampled points to points in $\R^{2n}$ and the distortion due to embedding in a lower dimensional subspace. For this, we need to analyze the effect of the input perturbation.  We prove the following result.
\begin{theorem}
Let $\mathscr{M}$ be a compact $k$-dimensional manifold in $\R^n$ having volume $\mathsf{V}$ and condition number $\kappa$. Let $\set{\bx_1, \cdots, \bx_m}$ be the points formed by finite data points of the underlying dimension $k$ as in~\figref{manifold}. Let $\mathsf{d}$ be the diameter of the datasets where the distance metric used to compute the diameter is the geodesic distance.
Let $\bPhi$ be a random projection that satisfies~\thmref{JL} such that 
\[ r \geq O \paren{\frac{k  \log (n \mathsf{V} \kappa) \log(1/\delta)}{ \varepsilon^{2}} }.  \]
Let $R$ and $R_{\bPhi}$ be the residual variance obtained when Isomap generates a $k$-dimensional embedding of $\mathbf{X}$ and $\bPhi \mathbf{X}$, respectively. Then we can construct a graph such that with probability at least $1-\delta$, 
\[ R_{\bPhi} \leq R + O \paren{c\varepsilon \mathsf{d}^2 + c\varepsilon\paren{ \frac{\ln(4/\beta)\sqrt{16 r \log (2/\beta)}}{\alpha}}^2 } .\]
\end{theorem}
In concise, the above theorem shows that we incur an additive error of $c\varepsilon \paren{ \frac{\ln(4/\beta)\sqrt{16 r \log (2/\beta)}}{\alpha}}^2$ in the residual variance due to the privacy requirement.
\begin{proof}
The proof proceeds by using an appropriate graph that uses the randomly projected sampled points as the vertices and then computes the bound on the overall residual variance incurred due to the errors in estimating the geodesic distances. The graph $G$ that we use is similar in spirit to the construction of Bernstein {\it et al.}~\cite{BDLT00}. Such graphs has the sampled points as its vertices.

For a graph $G$, we define a term {\em graph distance} between two vertices.  For two vertices $\mathbf{u},\mathbf{v}$ in the graph $G$, we denote the graph distance by $\| \mathbf{u} - \mathbf{v} \|_G$. The graph distance between two points $\mathbf{u}, \mathbf{v}$ is defined as 
\[ \| \mathbf{u} - \mathbf{v} \|_G := \min_P \set{ \|\mathbf{u}-\bx_{i_1} \|_2 + \cdots + \| \bx_{i_p}\mathbf{v} \|_2 }, \]
where the minimum is over all paths $P$ in the graph $G$  between $\mathbf{u}$ and $\mathbf{v}$ and the path $P$ has vertices $\mathbf{u}, \bx_{i_1}, \cdots, \bx_{i_p}, \mathbf{v}$. Using Theorem A of~\cite{BDLT00}, we have for some constant $\vartheta$, we have 
\[  (1-\vartheta) \| \mathbf{u} - \mathbf{v} \|_{\mathsf{Geo}} \leq  \| \mathbf{u} - \mathbf{v} \|_G \leq (1+ \vartheta)  \| \mathbf{u} - \mathbf{v} \|_{\mathsf{Geo}}. \]

Let $\widetilde{\bx}_1 = \bPhi \bx_1, \cdots, \widetilde{\bx}_m = \bPhi \bx_m$ be the points corresponding to the embedding. Let us denote the corresponding graph whose vertices are  the set $\set{\widetilde{\bx}_1, \cdots, \widetilde{\bx}_m}$ by $\bPhi G$. The edges of the graph is defined in a way so as to satisfy the requirements of Theorem A of~\cite{BDLT00} and suit our purpose as well. As in the case of the graph $G$, we define the distance on the new graph $\bPhi G$ formed by the vertices corresponding to the embedded points. 
\[ \| \mathbf{w} - \mathbf{z} \|_{\bPhi G} := \min_P \set{ \|\mathbf{w}-\widetilde{\bx}_{i_1} \|_2 + \cdots + \| \widetilde{\bx}_{i_p}\mathbf{z} \|_2 }. \]

In order to fully describe the graph $\bPhi G$, we need to define the edges present in the graph. We connect two nodes in $\bPhi G$ be an edge if $\| \mathbf{w} - \mathbf{z} \|_2 \leq (1+\eta) \varepsilon. $ It is important to note that we define the edges depending on the Euclidean distance. 

Now we use this construct of graph $\bPhi G$ to define another graph $G'$ whose vertices are the original sampled point. The graph $G'$ has an edge between two vertices only if there is an edge between the vertices corresponding to their embeddings, i.e., $\mathbf{u}$ and $\mathbf{v}$ are connected in $G'$ if and only if $\| \bPhi \mathbf{u} - \bPhi \mathbf{v} \|_2 \leq (1+\eta) \varepsilon. $ By our construction, $G'$ also meets the requirements for Theorem A of~\cite{BDLT00}. We can now finish the first part of the proof.

Let $P'$ be a path in $G'$ joining two vertices $\mathbf{u}$ and $\mathbf{v}$. Then by our construction, there is a path $P$ in $\bPhi G$ along the vertices corresponding to the embedding of the vertices along the path $P'$ in the graph $G'$. Let the path $P'$ consist of vertices $\mathbf{u}, \bx_{i_1}, \cdots, \bx_{i_p}, \mathbf{v}$. Then 
 \begin{align*}
  \| \widetilde{\mathbf{u}} - \widetilde{\bx}_{i_1} \|_2 + \cdots + \| \widetilde{\bx}_{i_p}-\widetilde{\mathbf{v}} \|_2   &= \| \bPhi {\mathbf{u}} -\bPhi {\bx}_{i_1} \|_2 + \cdots + \| \bPhi {\bx}_{i_p}- \bPhi {\mathbf{v}} \|_2 \\
  			&\leq (1 + \eta) \paren{ \|  {\mathbf{u}} - {\bx}_{i_1} \|_2 + \cdots + \|  {\bx}_{i_p}-  {\mathbf{v}} \|_2}. 
 \end{align*}
This holds for every path, and similar lower bound holds for every path. It therefore, follows that 
\[  (1-\eta) \| \mathbf{u} - \mathbf{v} \|_{G'} \leq  \| \mathbf{u} - \mathbf{v} \|_{\bPhi G}\leq (1+ \eta)  \| \mathbf{u} - \mathbf{v} \|_{G'}. \]

Using the result of~\cite{BDLT00}, we have the following.
\[  (1-\eta)(1-\vartheta) \| \mathbf{u} - \mathbf{v} \|_{\mathsf{Geo}} \leq  \| \mathbf{u} - \mathbf{v} \|_{\bPhi G}\leq (1+ \eta) (1+\vartheta)  \| \mathbf{u} - \mathbf{v} \|_{\mathsf{Geo}}. \]

Adjusting $\eta$ and $\vartheta$, we have the required bound on the geodesic distances. Therefore, the graph distances of $\bPhi G$ (which can be computed purely from the embedded points) provide a fairly good  approximation of the geodesic distance on the manifold.
 
For the second part, let $\bB$ be the matrix formed with entries $(a,b)$ representing the squared geodesic distances between $\bx_a$ and $\bx_b$. Let $\mathbf{I}$ be the corresponding matrix computed by the Isomap algorithm. From the analysis of Isomap, we know that the entries $(a,b)$ is
\[ \mathbf{I}_{ab} = - \frac{ \bB_{ab}^2 - \frac{1}{m}\sum_a \bB_{ab}^2 - \frac{1}{m} \sum_b  \bB_{ab}^2 + \frac{1}{m^2}\sum_{a,b} \bB_{ab}^2 }{2}.  \]

Further, $\mathbf{I} = \mathbf{J}^{\mathsf T} \mathbf{J}$, where $\mathbf{J}$ is a $k \times m$ matrix. In the Isomap algorithm, the final step computes the $k$-dimensional representation of every point by performing the eigenvalue decomposition of $\mathbf{I}$ by projecting column vectors of $\mathbf{X}$ onto the subspace spanned by  eigenvectors corresponding to the $k$ largest eigenvalues. The residual variance is defined as the sum of the $m-k$ eigenvalues of $\mathbf{I}$. In the ideal situation, when $\mathbf{I}$ is rank $k$ matrix, the residual variance is zero. However, we have some error due to projection and  perturbation of the sampled points. From~\thmref{geodesic}, we know that $\bx$ has a distortion of $\varepsilon$ or $3 \varepsilon $ in the case of squared norm. Moreover, any column $\bx$ of $\mathbf{X}$ is formed by distorting $\by$ of the sampled point, which is by an additive term of $\sigma_{\mathsf{min}}^2$. Therefore, every term in $\mathbf{I}$ suffers a distortion of at most $6 \varepsilon (\mathsf{d}^2 + \sigma_{\mathsf{min}}^2)$, where $\mathsf{d}$ is the square root of the largest entry of $\mathbf{B}$ (or the diameter of the sampled points). Therefore the matrix $\mathbf{I}$ varies as 
\[ \mathbf{I} (\varepsilon) = \mathbf{I} + 6 \varepsilon (\mathsf{d}^2 + \sigma_{\mathsf{min}}^2) \]
as a function of the distortion $\varepsilon.$ This implies that the eigenvalues can be  represented by the following equation:
\[ \lambda_i (\varepsilon) = \lambda_i + 6 \varepsilon (\mathsf{d}^2 + \sigma_{\mathsf{min}}^2) \mathbf{v}^{\mathsf T}_i \mathbf{E} \mathbf{v}_i, \]
where $\set{\mathbf{v}_i}_{i=1}^m$ are the eigenvectors and $\mathbf{E}$ is a matrix with entries $\pm 1$.  Hence the residual variance as a function of $\varepsilon$ can be written as :
\begin{align*}
	R_\bPhi = R(\varepsilon) &= \sum \lambda_i(\varepsilon) \\
				&= \paren{\sum_{i=k+1}^m \lambda_i} + 6 \varepsilon ({\mathsf d}^2+ \sigma_{\mathsf{min}}^2) \paren{\sum_{i=k+1}^n \mathbf{v}^{\mathsf T}_i \mathbf{E} \mathbf{v}_i} \\
				&= R + 6 \varepsilon ({\mathsf d}^2+ \sigma_{\mathsf{min}}^2) \paren{\sum_{i=k+1}^n \mathbf{v}^{\mathsf T}_i \mathbf{E} \mathbf{v}_i} \\
				& \leq R + 6 \varepsilon ({\mathsf d}^2+ \sigma_{\mathsf{min}}^2)  (m-k) \Lambda \\ 
				&< R + 6 m\varepsilon ({\mathsf d}^2+ \sigma_{\mathsf{min}}^2)   \Lambda,
\end{align*}where $\Lambda$ is the largest eigenvalue of $\mathbf{E}$.
Therefore, the average embedding distortion per sample point under the effect of random embedding and input perturbation (to preserve privacy) varies by at most $6\varepsilon ({\mathsf d}^2+ \sigma_{\mathsf{min}}^2)   \Lambda $. Note that $\Lambda$ is bounded by a constant as $\mathbf{E}$ is a bounded linear operator. This gives us the result of the theorem.
\end{proof}

The proof is now complete. \end{proof}
\end{versionC}
\begin{versionA}
\paragraph{Differential Privacy in the Streaming Model.}
\end{versionA}
\begin{versionC}
\subsection{Differentially private streaming algorithm}
\end{versionC}
Upadhyay~\cite{Upadhyay14} gave a method to convert a non-private streaming algorithm to a differentially private algorithm so that the resulting algorithm requires almost optimal space. However, the update time of their algorithm is slow by a linear factor in comparison to the non-private algorithms. Using~\thmref{DP} and combining with the results of Upadhyay~\cite{Upadhyay14}, we improve the update time of differentially private streaming algorithms to be the same as the non-private algorithms.
\begin{theorem}
	Given $\bA \in \R^{m\times n}, \mathbf{b} \in \R^n$. There is a privacy preserving data-structure that can be updated in time $O(m+\sqrt{m/\varepsilon} \log (1/\delta))$ and requires $O(m^2 \varepsilon^{-1} \log(1/\delta) \log (mn))$ bits   to solve linear regression of $\bA$ with respect to $\mathbf{b}$  in the turnstile model with $(\alpha,\beta+\delta)$-differential privacy.  The additive error incurred is $4{\sigma_\mathsf{min}\alpha} \| \bA^{\dagger} \mathbf{b} \|_F^2.$
\end{theorem}
Similar improvement for differentially private matrix multiplication of Upadhyay~\cite{Upadhyay14} is straightforward.
We give the description of these private algorithms  in the full version. 

\begin{versionC}
Formally, we have the following theorem.
\begin{theorem}
	Let $\bA$ be an $m_1 \times n$ matrix and $\mathbf{B}$ be an $m_2 \times n$ matrix. There is a privacy preserving data-structure to compute matrix multiplication $\bA \mathbf{B}^{\mathsf T}$  in the turnstile model that guarantees $(\alpha,\beta)$-differential privacy. The underlying data-structure used by the algorithm requires $O(m \varepsilon^{-2} \log(1/\delta) \log (mn))$ bits and  can be updated in time $O(m \log (\varepsilon^{-2}  \log (1/\delta)))$, where $m =\max \set{m_1,m_2}$. The additive error incurred is $\sigma_\mathsf{min}\alpha  \sqrt{N} \|\mathbf{B}\|_F$.
\end{theorem}
\end{versionC}
\begin{theorem} \label{thm:DP2}
	If the singular values of an $m \times n$ input matrix $\bA$ are at least $\frac{16 r \log (2/\beta)\ln(4/\beta)}{\alpha (1+\varepsilon)}$, then publishing $\bPhi^{\mathsf T} \bPhi \bA^{\mathsf T}$, where $\bPhi$ is as in~\figref{basic} preserves $(\alpha, \beta+\delta )$-differential privacy. 
\end{theorem}
\begin{versionA}
We present a proof of this theorem in the full version. 
\end{versionA}

The proof of the above theorem is fairly straightforward by invoking the second main result of Upadhyay~\cite{Upadhyay14} which is as follows.
\begin{theorem} \label{thm:PSG2} {\em (Upadhyay~\cite{Upadhyay14})}.
	Let $\mathbf{R}$ be a random matrix such that every entry of $\mathbf{R}$ is sampled independently and identically from $\cN(0,1)$. If the singular values of an  input matrix $\mathbf{C}$ are at least $\frac{16 r \log (2/\beta)\ln(4/\beta)}{\alpha}$, then publishing $\mathbf{R}^{\mathsf T}\mathbf{R} \mathbf{C}^{\mathsf T}$ is $(\alpha, \beta)$-differentially private. 
\end{theorem}

\begin{proof}
Substitute $\mathbf{C}=\mathbf{B} \bA^{\mathsf T}$ and $\mathbf{R}=\mathbf{P}_1$, where $\bA$ is as in~\thmref{DP2}, $\mathbf{B}$ is as in~\lemref{singular}, and $\mathbf{P}_1$ be as in~\secref{newJL}.
\end{proof}

\begin{versionC}
This leads to the following theorem.
\begin{theorem}
	Let $\lambda_1\geq \cdots \geq \lambda_{{\sf rk}(\bA)}$ be the singular values of $\bA$. Then for an over-sampling parameter $p$, there is a single-pass mechanism that compute the $k$-rank approximation $\bar{\bA}$ using  $O(k(n+d) \alpha^{-1} \log(nm))$ bits while preserving $(\alpha, \beta+\delta)$-differential privacy such that	
	\begin{align*} 
	(i) \quad \| A - \bar{\bA} \|_F  & \leq 
	 \paren{1 + \frac{k}{p-1}}^{1/2} \min_{rk(A')<k}\|\bA-\bA' \|_F + \frac{ 2k}{\alpha(1+\varepsilon)}\sqrt{\frac{(n+d) \ln (k/\beta) }{p}} , \qquad \text{and} \\
	(ii) \quad \| \bA - \bar{\bA} \|_2  & \leq 
	\paren{1 + \frac{k}{p-1}}^{1/2} \lambda_{k+1}  + \frac{e\sqrt{(k+p) \sum_{j>k} \lambda_j^2}}{p} + \frac{ 2\sqrt{k(n+d) \ln (k/\beta) }}{\alpha (1+\varepsilon)} .
	\end{align*} 
\end{theorem}
\begin{proof}
The privacy proof follows by plugging in~\thmref{DP2} in the privacy proof of~\cite{Upadhyay14}. The utility proof follows in the same way as in~\cite{Upadhyay14}.
\end{proof}

\section{Other Applications} 
\subsection{Compressed Sensing} \label{sec:compressed}
\end{versionC}
\begin{versionA}
\section{Application in Compressed Sensing and More} \label{sec:compressed}
\end{versionA}
Our bound in~\thmref{newsparse} directly translates to give a random matrix with optimal $\mathsf{RIP}$ for $r \leq n^{1/2-\tau}$.  
\begin{versionA}
More concretely, combining the result of Baraniuk et al~\cite{BDDW07} stated in~\thmref{BDDW07} with~\thmref{newsparse}, we have.
\end{versionA}
\begin{versionC}
Baranuik {\it et al.}~\cite{BDDW07} showed the following:
\begin{theorem} \label{thm:BDDW07}
	Let $\bPhi$ be a $r \times n$ random matrix  drawn according to any distribution that satisfies~\eqnref{JL}. Then, for any set $T$ with $|T|=s<n$ and any $0\leq \delta \leq 1$, we have $ \forall x \in X_T$
	\[ \p_\bPhi [ (1-\delta) \| \bx \|_2 \leq \| \bPhi_T \bx \|_2 \leq (1+\delta) \| \bx \}_2] \geq 1 - 2(12/\varepsilon)^k e^{-O(\delta n)},\]  where $X_T$ is all $\bx \in X$ restricted to indices in the set $T$.
	Moreover,  the $\mathsf{RIP}$ holds for $\bPhi$ with the prescribed $\delta$ for any $s \leq O\paren{ {\delta^2r}/{\log (n/r)}} $ with probability at least $1- 2e^{-O(r)}$.
\end{theorem} 
Baranuik et al~\cite{BDDW07}  argued that their bound is optimal. 
To quote, 
	``From the validity of theorem (\thmref{BDDW07}) for the range of $s \leq c r /\log(n/r)$, one can easily deduce its validity for $k \leq c' r/(\sparen{\log(n/r)+1})$ for $c'>0$  depending only on $c.$"
Combining~\thmref{BDDW07} with~\thmref{newgaussian}, we have,
\end{versionC}
\begin{theorem}
	For any set $T$ with $|T| \leq s \leq O\paren{ \frac{r}{\log (n/r)}}$, there exists a mapping $\bPhi$ such that, for all $\bx \in X$ with non-zero entries at indices in the set $T$, we have 
	$\p \sparen{ (1-\delta) \| \bx \|_2 \leq \| \bPhi \bx \|_2 \leq (1+\delta) \| \bx \|_2} \geq 1 - 2(12/\delta)^k e^{-O(\delta n)}$. Moreover, $\bPhi$ can be sampled using $2n + n \log n$ random bits and $\bPhi \bx$ can be computed in time $O(n \log b)$. Here $ b= \min \{c_2 a \log\paren{\frac{r}{\delta}},r\}$ with $a=c_1 \varepsilon^{-1} \log\paren{\frac{1}{\delta}} \log^2\paren{\frac{r}{\delta}}$.
\end{theorem}

\begin{versionA}
\begin{remark}
Our distribution can be also applied in other applications of {\sf JLP}, including learning theory, fast numerical algebra, and functional analysis. Due to lack of space, we defer their details in the full version. 
\end{remark}
\end{versionA}
\begin{versionC}
\subsection{Functional Analysis} \label{sec:functional}
Charikar and Sahai~\cite{CS02} showed an impossibility result for dimension reduction using linear mapping in $\ell_1$ norm. Johnson and Naor~\cite{JN10} gave an elegant result which generalizes Charikar and Sahai's result. Concisely, they showed that a normed space which satisfies the {\sf JL} lemma is very close to being Euclidean in the sense that all its $n$-dimensional subspaces are isomorphic to Hilbert space with distortion $2^{2^{O(\log^* n)}}$. 
More formally, they showed the following.
\begin{theorem} \label{thm:JN10}
Let $\cX$ be a Banach space such that for every $n$ and every $\bx_1,\cdots , \bx_m \in \cX$, there exists a 
linear subspace $\mathbf{F} \in \cX$ of dimension at most $O(\varepsilon^{-2} \log n)$ and a linear mapping from $\cX$ to $\mathbf{F}$
 such that~\eqnref{JL} is satisfied. Then for every $\ell \in \N$ and every $\ell$-dimensional subspace $\cE \subseteq \cX$, we have the Banach-Mazur distance between $\cE$ and Hilbert space  at most $2^{2^{O(\log^* \ell)}}$.
\end{theorem}

Our main result directly improves this characterization in the sense that we only need almost linear random samples. However, there are another subtle applications of our result in Functional analysis. For this, we need the result by Krahmer and Ward~\cite[Proposition 3.2]{KW11}. 
{\begin{theorem} {\em (Krahmer-Ward~\cite[Proposition 3.2]{KW11})} \label{thm:KW11}
	Let $\varepsilon$ be an arbitrary constant. Let $\bPhi$ be a matrix of order $k$ and dimension $r \times n$ that satisfies the relation $k \leq c_1 \delta_k^2 r / \log (n/r)$. Then the matrix $\bPhi \bD$, where $\bD$ is an $n \times n$ diagonal matrix  formed by Rademacher sequence, satisfies~\eqnref{JL}.
\end{theorem} 
}
Recently, Allen-Zhu, Gelashvili, and Razenshteyn~\cite{AGR14} showed that there does not exists a projection matrix with $\mathsf{RIP}$ of optimal order for any $\ell_p$ norm other than $p=2$. Using the characterization by Johnson and Naor, our result gives an alternate reasoning for the impossibility result of~\cite{AGR14}. 
Moreover, Pisier~\cite{Pisier} proved that if a Banach space satisfies the condition that Banach-Mazur distance is $o(\log n)$, then the Banach space is super-reflexive. Because of the uniform convexity properties of the norm over these spaces~\cite{Finet86}, super-reflexive Banach spaces are of centre importance in Functional analysis (see for example the excellent book by Conway~\cite{Conway}). \thmref{newsparse} and~\thmref{JN10} gives one such space using only almost linear random samples! 

\subsection{Numerical Linear Algebra}
Our bound gives the same results as obtained by Kane and Nelson~\cite{KN14} for matrix multiplication, which is optimal as long as $m \leq 2^{\sqrt{n}}$. For the case of linear regression, we achieve the same bound as long as $m \leq \sqrt{n}$. This is because of a subtle reason. The approximation bound for linear regression uses an observation made by Sarlos~\cite{Sarlos06}, who observed that the approximation guarantee for linear regression  is provided as long as $\bPhi$ provides good approximation to matrix multiplication as long as $\bPhi$ is an $O(m)$-space embedding. The latter follows from the standard result on any Johnson-Lindenstrauss transform with $r=O(\mathsf{rank} (\bA) + \log (1/\delta) /\varepsilon^2)$. Using the result of Indyk and Motwani~\cite{IM98}, we only need $r=O(m^2\varepsilon^{-1} \log(1/\delta))$. In concise, we need to project to a subspace with dimension $r=O(m/ \varepsilon)$. Using our distribution over random matrices, this is possible only if $m \leq O( \sqrt{n}/\varepsilon).$

\end{versionC}

\section{Impossibility of Differential Privacy Using Other Known Constructions} \label{sec:impossible} \label{sec:notDP}
Blocki {\it et al.}~\cite{BBDS12} asked the question whether known distribution of matrices with {\sf JLP} preserves differential privacy or not. In this section, we answer their question in negation by giving simple counter-examples. Our counter-example works by showing two neighbouring matrices with the properties that they differ by at most one column of unit norm. Each of our counter-examples are constructed with care so as to allow the attack. 
\begin{theorem} \label{thm:notDP}
Let $\bPhi$ be a random matrix picked from a distribution defined by Nelson {\it et al.}~\cite{NPW14}, Kane and Nelson~\cite{KN14}, Dasgupta {\it et al.}~\cite{DKS10}, or Ailon and Liberty~\cite{AL09}.  Then the distribution  $\bPhi \bA$ and $\bPhi \bA'$ are distinguishable with probability $1/\poly(n)$, where $\bA$ and $\bA'$ are  two neighbouring matrices.
\end{theorem}
We prove the above theorem by giving counter-examples designed specifically for each of the above constructions. 
\begin{versionA}
Due to lack of space, we defer the details to~\appref{notDP}.
\end{versionA}
\begin{versionC}
\subsection{Random Linear Combination Based Construction of Nelson, Price, and Wooters~\cite{NPW14}} \label{app:NPW14} \label{sec:NPW14}
We give verbatim the description of the Nelson {\it et al.}~\cite{NPW14} of their construction of {\sf RIP} matrices. The construction for matrices with {\sf JLP} follows straightforwardly from the result of Krahmer and Ward~\cite{KW11}.
\begin{quote}
	Let $\cD_M$ be a distribution on $M \times n$ matrices, defined for all $M$, and fix parameters $r$ and $B$. Define the injective function $h: [B] \times [r] \rightarrow [rB]$ as $h(b, i) = B(b-1)+i$ to partition $[rB]$ into $r$ buckets of size $B$, so $h(b, i)$ denotes the $i$-th element in bucket $b$. We draw a matrix  $\bPhi' \sim \cD_{rB}$, and then construct our $r \times n$ matrix $\bPhi$ by using $h$ to hash the rows of $\bPhi'$ into $r$ buckets of size $B$. Define a distribution on $r \times n$ matrices by constructing a matrix $\bPhi$ as follows:
	\begin{enumerate}
		\item Draw $\bPhi' \sim \cD_{rB}$, and let $\bPhi'_{i:}$ denote the $i$-th row of $\bPhi'$.
		\item For each $(b,i) \in [m] \times [B]$, chose a sign $\sigma_{b,i} \sim \sber(1/2) $.
		\item For $1 \leq b \leq r$, let
			\[ \bPhi_{b:} := \sum_{i \in [B]} \sigma_{b,i} \bPhi'_{h(b,i)} \]
			and let $\bPhi$ be the matrix with rows $\bPhi_{b:}$.
	\end{enumerate}
\end{quote}

First note that with probability $2^{-(B+1)}$,  $\sigma_{1,i}=1$ and $\sigma_{2,i}=1$ for all $1 \leq i \leq B$. In other words, $\bPhi_{1:} = \sum_{i \in [B]}  \bPhi'_{h(1,i)}$ and $\bPhi_{2:} = \sum_{i \in [B]}  \bPhi'_{h(2,i)}$. In the counter-example for both the construction by Nelson {\it et al.}~\cite{NPW14}, we concentrate on the the top-most left-most $2 \times 2$ sub-matrix of $\bPhi$. We denote this sub-matrix by $\bPhi_{22}$ for brevity. Recall that $B=\poly \log n$ in both the constructions of Nelson {\it et al.}~\cite{NPW14}.

\paragraph{Construction Based on Bounded Orthonormal Matrices.} 
Recall that the two neighbouring matrices used in our counterexample in the introduction is as follows.
\[
\bA= 	\begin{pmatrix} w & 1 & 0 & \cdots & 0 \\  
				0 & w & 0 & \cdots & 0 \\
  				0  & 0 & w & \cdots & 0\\
  				\vdots & \vdots & & \ddots \\
  				\mathbf{0} & 0 & 0 & \cdots & w
			\end{pmatrix} \quad	\text{and} \quad
\widetilde{\bA}= 	\begin{pmatrix} w & 1 & 0 & \cdots & 0 \\  
				1 & w & 0 & \cdots & 0 \\
  				0  & 0 & w & \cdots & 0\\
  				\vdots & \vdots & & \ddots \\
  				\mathbf{0} & 0 & 0 & \cdots & w
			\end{pmatrix} 	
		 \]   for $w \geq \sigma_{\mathsf{min}}.$ 
They are neighbouring because the two matrices differ only in the first column and their difference have a unit norm, $\|\bA - \widetilde{\bA}\|_2 =1.$ We consider the $2 \times 2$ sub-matrices of $\bA$ and $\widetilde{\bA}$ to show that this construction does not preserve differential privacy. The two sub-matrices have the following form:
\begin{align}
\begin{pmatrix} w & 1 \\ 0 & w   \end{pmatrix} \quad  \text{and} \quad
\begin{pmatrix} w & 1 \\ 1 & w   \end{pmatrix}. 
\label{eq:neighbor}
\end{align}

Now considering $\bPhi_{22}$, the possible (non-normalized) entries are as follows:
\[
\begin{pmatrix} 1 & 1 \\ 1 & -1   \end{pmatrix}, \quad 
\begin{pmatrix} 1 & 1 \\ 1 & 1   \end{pmatrix}, \quad 
\begin{pmatrix} 1 & -1 \\ 1 & -1   \end{pmatrix}, \quad 
\begin{pmatrix} 1 & 1 \\ 1 & 0   \end{pmatrix}, \quad 
\begin{pmatrix} 1 & 0 \\ 1 & 1   \end{pmatrix}. 
\]

Now, with probability $2^{-B-3} = 1/\poly(n)$, the corresponding $2 \times 2$ sub-matrix formed by multiplying $\bPhi \bD \bA$ are
\[
\begin{pmatrix} w & w+1 \\ w & 1-w   \end{pmatrix}, \quad 
\begin{pmatrix} w & w+1 \\ w & w+1   \end{pmatrix}, \quad 
\begin{pmatrix} w & 1-w \\ w & 1-w   \end{pmatrix}, \quad 
\begin{pmatrix} w & w \\ w & 1   \end{pmatrix}, \quad 
\begin{pmatrix} w & 1 \\ w & w+1   \end{pmatrix}. 
\]
On the other hand, the corresponding $2 \times 2$ sub-matrix formed by multiplying $\bPhi \bD \widetilde{\bA}$ are
\[
\begin{pmatrix} w+1 & w+1 \\ w-1 & 1-w   \end{pmatrix}, \quad 
\begin{pmatrix} w+1 & w+1 \\ w+1 & w+1   \end{pmatrix}, \quad 
\begin{pmatrix} w-1 & 1-w \\ w-1 & 1-w   \end{pmatrix}, \quad 
\begin{pmatrix} w+1 & w+1 \\ w & 1   \end{pmatrix}, \quad 
\begin{pmatrix} w & w+1 \\ w+1 & w+1   \end{pmatrix}. 
\]

Note that in each of the cases, an adversary can distinguish the two output distribution. For example, in the first case, the entries in the second row for $\widetilde{\bA}$ is negation of each other while it is not so in the case of $\bA$. Similarly, in the second case, all entries are the same in $\widetilde{\bA}$ while different in $\bA$.

Note that publishing $\bPhi \bA^{\mathsf T}$ when $w \geq \sigma_{\mathsf{min}}$ preserves privacy when $\bPhi$ is a random Gaussian matrix. 

\paragraph{Hash and Partial Circulant Matrices Based Construction.} Let $\mathbf{v}=\begin{pmatrix} 1/\sqrt{2} & 1/\sqrt{2} & 0 & \cdots &0 \end{pmatrix}^{\mathsf T}$. 
In this case, we use the following set of neighbouring matrices:
\begin{align}
A = w \mathbf{F}_n^{\mathsf T} \quad \text{and} \quad 
A' = w \mathbf{F}_n^{\mathsf T} +  \mathbf{v} \mathbf{e}_2^{\mathsf T}
\label{eq:partial}
\end{align}
  for $w \geq \sigma_{\mathsf{min}}.$  They are neighbouring because $\|\bA - \widetilde{\bA}\|_2 =1$ and they differ in only one column. As in the previous two cases, we  concentrate on the $2 \times 2$ sub-matrices, which would be enough to argue that an adversary can distinguish the output distribution corresponding to these two neighbouring matrices.  The partial circulant matrix can be represented as  $\bP = \mathbf{F}_{1..r}^{\mathsf T} \mathsf{Diag}(\mathbf{F}_n \mathbf{g}) \mathbf{F}_n$, where $\mathbf{F}_{1..r}^{\mathsf T}$ is the restriction of $\mathbf{F}_{n}^{\mathsf T}$ to the first $r$-rows and $\mathsf{Diag}(\mathbf{F}_n \mathbf{g})$ is the diagonal matrix with diagonal entries formed by the vector $\mathbf{F}_n \mathbf{g}$.  Let $\by = \mathbf{F}_n \mathbf{g}$, then the $2 \times 2$ sub-matrix formed by the application of $\mathsf{Diag}(\by) \mathbf{F}_n \bD \bA$ and  $\mathsf{Diag}(\by) \mathbf{F}_n \bD \widetilde{\bA}$ yields,
\begin{align}
\begin{pmatrix}
 w \by_1 & 0 \\ 0 & w \by_2 
\end{pmatrix} 
\quad \text{and} \quad 
\begin{pmatrix}
 w \by_1 &  \sqrt{2} \by_1 \\ 0 & w \by_2 ,
\end{pmatrix} \label{eq:fourier}
\end{align}
respectively. The rest of the computation is similar to the case of bounded-orthonormal matrices after observing that the two sub-matrices have the form quiet similar to~\eqnref{neighbor} and that the rest of the steps are the same as in the bounded-orthonormal case. An alternative could be the following line. We  multiply the published matrix with $\mathbf{F}_{1..r}$. In other words, we need to concern ourselves with $\mathsf{Diag}(\mathbf{F}_n \mathbf{g}) \mathbf{F}_n \bD \bA$ and  $\mathsf{Diag}(\mathbf{F}_n \mathbf{g}) \mathbf{F}_n \bD \widetilde{\bA}$. Thus~\eqnref{fourier} suffices to distinguish.

Note that publishing $\bPhi \bA^{\mathsf T}$ when $w \geq \sigma_{\mathsf{min}}$ preserves privacy when $\bPhi$ is a random Gaussian matrix. 

\subsection{Hashing Based Construction}
There are two main constructions that were based on hashing co-ordinates of the input vectors to a range $\{1, \cdots, S\}$~\cite{DKS10,KN14}, where $S$ is the sparsity constant. The main difference between the two construction is that Dasgupta, Kumar, and Sarlos~\cite{DKS10} hash co-ordinates with replacement while Kane and Nelson~\cite{KN14} hash them without replacement. More details on their construction and the counterexample follows.
\subsubsection{Construction of Dasgupta, Kumar, and Sarlos~\cite{DKS10}}
We first recall the construction of Dasgupta {\it et al.}~\cite{DKS10}. Let $a=c_1 \varepsilon^{-1} \log(1/\delta) \log^2(r/\delta)$, $ b= \min \{c_2 a \log(a/\delta),n\}$. Let $\bD = \set{d_1, \cdots, d_{ab} }$ be a set of i.i.d. random variables such that for each $j \in [ab]$, $d_j \sim \sber(1/2)$. Let $\delta_{\alpha, \beta}$ denote the Kronecker delta function. Let $h : [b] \rightarrow [r]$ be a hash function chosen uniformly at random. The entries of matrix $\mathbf{H} \in \{0,\pm 1\}^{r \times b}$ is defined as $\mathbf{H}_{ij} = \delta_{i,h(i)} d_j$. Then their random matrix is $\bPhi:= \mathbf{H}  \bW \bD$. 

In order to analyze this construction, we need to understand the working of the random matrix $\mathbf{H}$. The main observation here is that the matrix $\mathbf{H}$  has column sparsity $1$, i.e., every column has exactly one entry which is $\pm 1$ and rest of the entries are zero. Let $j$ be the row for which  the first column is non-zero and $k$ be the row for which the second column is non-zero.

Our neighbouring data-sets are exactly like in~\eqnref{partial} except that we use  Hadamard matrix instead of the discrete Fourier transform.
\begin{align}
\bA = w \bW^{\mathsf T} \quad \text{and} \quad 
\widetilde{\bA} = w \bW^{\mathsf T} +  \mathbf{v} \mathbf{e}_2^{\mathsf T}
\label{eq:partial2}
\end{align}
  for $w \geq \sigma_{\mathsf{min}}.$  They are neighbouring because $\|\bA - \widetilde{\bA}\|_2 =1$ and they differ in only one column.
Now just as before, with probability $1/4$, the first $2 \times 2$ sub matrix formed by $\bW \bD \bA$ and $\bW \bD A'$ is as below
\[
\begin{pmatrix}
 w  & 0 \\ 0 & w  
\end{pmatrix} 
\quad \text{and} \quad 
\begin{pmatrix}
 w  & \sqrt{2} \\ 0 & w  
\end{pmatrix} ,
\]
respectively. Now  consider the sub-matrix formed by the $j$-th and $k$-th row and the first and second columns  of the output matrices corresponding to $\bA$ and $\widetilde{\bA}$. Just like before, simple calculation shows that in each of the possible cases, the two distribution is  easily distinguishable. This breaches the privacy.

Note that publishing $\bPhi \bA^{\mathsf T}$ when $w \geq \sigma_{\mathsf{min}}$ preserves privacy when $\bPhi$ is a random Gaussian matrix. 

\begin{remark}
A remark is due here. Note that the above attack works because we know explicitly what would be the output of the hash function $h$. If the hash function outputs a value picked i.i.d. from a Gaussian distribution, we could not argue as above. Moreover, since the Gaussian distribution is subgaussian, the analysis of Kane and Nelson~\cite{KN14} would still hold and the resulting distribution of random matrices with samples picked from $\sber(1/2)$ replaced by a Gaussian variable would also satisfy the {\sf JLP}.  
\end{remark}

\subsubsection{Construction of Kane and Nelson~\cite{KN14}}
The first two constructions of Kane and Nelson~\cite{KN14} differs from that of Dasgupta, Kumar, and Sarlos~\cite{DKS10}  in the way they pick which co-ordinates of the input vector should be multiplied with the Rademacher variable. Both of them employ hash functions to perform this step. In Kane and Nelson~\cite{KN14}, hashing is done without replacement, whereas, in Dasgupta, Kumar, and Sarlos~\cite{DKS10}, hashing is done with replacement. Since our counterexample does not depend whether or not the hashing is done with replacement, the counterexample we exhibited in the case of the construction of Dasgupta, Kumar, and Sarlos~\cite{DKS10} is also a counter-example for the construction of Kane and Nelson~\cite{KN14}. 

\subsection{Construction of Ailon and Liberty~\cite{AL09}}
The construction of Ailon and Liberty~\cite{AL09} has the following form $\bPhi = \bB \bW \bD$, where $\bB$ is a random dual BCH code and $\bD$ is a diagonal matrix as in all the other constructions. This can be equivalent to saying that $\bPhi $ has the following form:
\[ \bPhi = \sqrt{\frac{1}{r}} \mathbf{H} \bW \bD^{(1)} \bW \bD^{(2)} \bW \cdots \bW \bD^{(\ell)}, \]
where $\bD^{(i)}$ are independent diagonal matrices with non-zero entries picked i.i.d. from $\sber(1/2)$ and $\mathbf{H}$ is a subsampled (and rescaled) Hadamard-transform. 
They proved that this construction satisfies the {\sf JLP} for 
the regime $c \nu^{-2}s\log n \leq r \leq n^{1/2-\tau}$ with $\ell=\floor{-\log(2\sqrt{n/r})/\log k}$ for $k <1/2$  and  for $c \nu^{-2}s\log n \leq r \leq n^{1/3}/\log^{2/3} n$ with $\ell=3$. 
In what follows, we show a simple attack when $\ell=3$. The case for arbitrary $\ell=\floor{-\log(2\sqrt{n/r})/\log k}$ follows similar ideas and is omitted. Note that because of the choice of $\ell$, this leads to an at most polynomial degradation of our analysis.

As before, we give two neighbouring data-sets as in~\eqnref{partial2}.
\begin{align}
\bA = w (\bW^{\mathsf T})^3 \quad \text{and} \quad 
\widetilde{\bA} = w (\bW^{\mathsf T})^3 + \mathbf{v} \mathbf{e}_2^{\mathsf T}
\end{align}
  for $w \geq \sigma_{\mathsf{min}}.$  They are neighbouring because $\|\bA - \widetilde{\bA}\|_2 =1$ and they differ in only one column.
The probability that these two rows are picked is $2r/n$. Since the probability that each of the $\bD^{(j)}$ has the first two non-zero entries $1$ is $1/8$. Therefore, the same conclusion as before holds with probability $r/4n$.

Note that publishing $\bPhi \bA^{\mathsf T}$ when $w \geq \sigma_{\mathsf{min}}$ preserves privacy when $\bPhi$ is a random Gaussian matrix. 

\end{versionC}

\section{Analysis of Previously Known Distributions with {\sf JLP}} \label{sec:older} \label{sec:others}
\begin{versionA}
In this section, extending on the work by Kane and Nelson~\cite{KN10,KN14}, we give simpler proof of few other constructions, including the variants of the Johnson-Lindenstrauss transform based on~\cite{AC09,Matousek08} and bounded moment construction of Arriaga and Vempala~\cite{AV06}. The central idea in all these proofs is to cast the problem of proving isometry in the terms of~\thmref{hw}.  This gives a more unified picture of the Johnson-Lindenstrauss transform. Due to lack of space, we defer the details to~\appref{analysis}. We also give a simplified self-contained proof of~\thmref{hw} instantiated with Gaussian random variables in~\appref{gaussian}; thereby, making the proof of distributions using Gaussian variables self-contained. A simplified proof of~\thmref{hw} instantiated with $\sber(1/2)$ already exists in literature.
\end{versionA}
\begin{versionC}
We conclude this paper by analysing and revisiting some of the older known constructions. A simpler proof of few of the known constructions have been shown by Kane and Nelson in two successive works~\cite{KN10,KN14}.
In this section, we give altnerate analysis of few known distribution over random matrices satisfying {\sf JLP}. Few of these results have been shown by Kane and Nelson in two successive works~\cite{KN10,KN14}. At the centre of our proofs is~\thmref{hw}. To be self-contained, we give a simple proof of~\thmref{hw} with Gaussian random variables in~\appref{gaussian}.

\paragraph{Random Matrix Based Construction.} The first construction proposed by Johnson and Lindenstrauss~\cite{JL84} uses a random Gaussian matrix as the projection matrix $\bPhi$, i.e., each entry $\bPhi_{ij} \sim \cN(0,1)$ for $1\leq i \leq r, 1 \leq j \leq n$. Noting that the result of~\thmref{hw} does not involve the dimension of the matrix $\bA$ and subgaussian random vectors, the proof follows simply from~\lemref{gaussian} and~\thmref{hw} by substituting every $\bz_i=\bx$ to be the same unit vector $\bx$.


Kane and Nelson~\cite{KN10} showed that the distribution of Achlioptas~\cite{Achlioptas03}  can be shown to satisfy {\sf JLP} using the Hanson-Wright's inequality. We just restated their result.
This construction has every entries of $\bPhi$ picked i.i.d. from $\sber(1/2)$. Since this distribution is also subgaussian, the proof  follows by invoking~\thmref{hw} with a vector whose entries are picked i.i.d. from $\sber(1/2)$. 

The second construction of Achlioptas~\cite{Achlioptas03} is aimed to reduce the number of non-zero entries. It is defined by the following distribution:
\begin{align}
\bPhi_{ij} := \left\{ \begin{matrix} +\sqrt{3} & \text{with probability}~1/6 \\
									0 & \text{with probability}~2/3 \\
								   -\sqrt{3} & \text{with probability}~1/6 
\end{matrix} \right. \label{eq:achlioptas}
\end{align}
It is easy to check from the moment-generating function that the above distribution is also sub-gaussian. The proof thus follows as above.

\paragraph{Bounded Moment Construction.} Arriaga and  Vempala~\cite{AV06} showed that any bounded random variables  satisfies the Johnson-Lindenstrauss transform. The proof using Hanson-Wright now follows from the following lemma which states that random variables with bounded random variables are sub-gaussian. 
\begin{lemma}
Suppose $X$ is a symmetric random variable which is bounded. Then $X$ is subgaussian.
\end{lemma}
\begin{proof}
Suppose $X$ is a bounded random variable with a symmmetric
distribution. That is, $| X| \leq M$ for some constant $M$ and $-X$ has the same
distribution as $X$. Then
\[  \p [\exp(t X)]=1+\sum_{k\geq 1}\frac{t^k  \p [X^k] } {k !}. \]

By symmetry, $\p [X^k] = 0$ for each odd $k$. For even $k$,  $\p [X^k] \leq M^k$ , leaving
\[  \p [\exp(t X)]=1+\sum_{k\geq 1}\frac{t^{2k}  \p [X^{2k}] } {(2k)!} \leq \exp(M^2 t^2/2). \]
because $(2k)! \geq 2^k k!$ for each $k$.
This completes the proof because this is one of the properties of subgaussian distribution.
\end{proof}

\paragraph{Fast Johnson-Lindenstrauss transform and variants.} The projection matrix used by Ailon and Chazelle~\cite{AC09} is $\bPhi=\bP \bW \bD$, where $\bW$ is the  Hadamard matrix, $\bD$ is a diagonal matrix with non-zero entries picked i.i.d. from $\sber(1/2)$, and the matrix $\bP$ is constructed in the following manner:  every entry $\bP_{ij}$ is picked from normal distribution $\cN(0,1)$ with probability $p$ and set to $0$ with probability $1-p$.  The parameter $p$ governs the number of non-zero entries in $\bP$. They set $p=\Theta(n^{-1}\log^2 (1/\delta))$.

We  look at this construction in an alternate way, which was used by Upadhyay~\cite{Upadhyay13}. The matrix $\bP$ can be seen as formed by the composition of two matrices: an ($r \times t$) matrix $\bP_1$ and a ($t \times n$) matrix $\bP_2$ in the following way. $\bP_1$ is a random Gaussian matrix with entries picked i.i.d. from $\cN(0,1)$ and $\bP_2$ randomly samples co-ordinates from $\sqrt{n/t}(\bW \bD \bx)$ with replacement. From~\lemref{gaussian}, we know that $\bP_1$ preserves isometry upto a multiplicative factor of $(1 \pm \varepsilon)$. We prove the same about $\bP_2$. For that we use~\thmref{AC09} which states that $\| \bW \bD\bx \|_\infty \leq \sqrt{\log(n/\delta)}/n$ for $\bx \in \bS^{n-1}$. 

Let $\by$ be the vector formed by sampling due to $\bP_2$ sampling co-ordinates of $\sqrt{n/t}(W D \bx)$ with replacement. Therefore, $\|\by\|^2_2 = \sum_{i=1}^t \by_i^2$. Also, $\E[\|\by\|^2]=1$. Hoeffding's bound gives the following
\begin{align}\p \sparen{ | \|\by\|_2^2 - 1 | > \varepsilon } \leq \exp(-\varepsilon^2 t / \log(n/\delta)). \end{align}
For this to be less than $\delta$, pick $t = \varepsilon^{-2} \log(1/\delta) \log(n/\delta).$ This was the value chosen by Ailon and Chazelle~\cite{AC09} as well.

The construction of Ailon and Liberty~\cite{AL09} can also be shown to satisfy {\sf JLP} once we notice that subsampling $\bW \bD$ results in a matrix which has  $2 \floor{\log (1/\delta)}$-wise independent entries. We can therefore invoke~\thmref{hw} as it is true when the random variables picked from a sub-gaussian distribution are $2 \floor{\log (1/\delta)}$-wise independent. This was known as folklore, and, to our knowledge, shown first by Kane and Nelson~\cite{KN14}.
For the construction of Matousek~\cite{Matousek08}, we make the following remark.
\begin{remark}
	We remark that since signed Bernoulli random variable is sub-gaussian, we can replace $\bP$ by the following matrix: with probability $p$, sample from $\sber(1/2)$; otherwise, set it to $0$. This was exactly the construction of Matousek~\cite{Matousek08}. 
\end{remark}


\paragraph{Sparse Johnson-Lindenstrauss Transform.} Kane and Nelson~\cite{KN10} have shown the proof of the sparse-JL construction of Dasgupta {\it et al.}~\cite{DKS10} using the Hanson-Wright inequality. Kane and Nelson~\cite{KN14} gave a code-based construction which they proved preserves Euclidean norm (see Section 3 of their paper). They also used~\thmref{hw} invoked with $\sber(1/2)$. We do not know if their main construction with optimal sparsity constant can be simplified using~\thmref{hw}. To our knowledge, these three are the only known constructions of sparse-Johnson Lindenstrauss transform.

\end{versionC}

\begin{versionC}
\section{Discussion and Open Problems} \label{sec:open} \label{app:open}
In this paper, we explored the properties of a distribution of random matrices  constructed using linear random samples. We gave a distribution of random matrices that satisfies {\sf JLP} as long as $r = O(n^{1/2-\tau})$. The reason for this bound comes from~\lemref{P_2}. This is the single source of slackness and could lead to a  belief that any improvement in this lemma would result in overall improvement of the bound. We now show that we cannot improve on it with the present construction. The issue that we face was also faced by Ailon and Liberty~\cite{AL09}. We follow up with more details.

  Let $\bx \in \bS^{n-1}$ which has $1/n^{1/4}$ on the first $\sqrt{n}$ coordinates and is $0$ everywhere else. With probability $2^{-\sqrt{n}}$, $\bD$ does nothing to $\bx$ (i.e. $\bD \bx = \bx$). Then $\bW \bx$ is just some other vector with equal mass spread on $\sqrt{n}$ coordinates, and so does when we apply $\bPi$. So with probability at least $2^{-\sqrt{n}}$, $\bP$ is just being applied to some permutation of $\bx$. In the best case scenario none of the coordinates of x land in the same block, in which case $\| \bP \bx\|_2^2 = (1/t) \sum_{i=1}^t \mathbf{g}_i^2$ for independent gaussians $\mathbf{g}_i$, where $t = \sqrt{n}$.  For the expression $(1/t) \sum_{i=1}^t \mathbf{g}_i^2$ to be an isometry  with probability $1-\delta$, we need $t \approx \varepsilon^{-2} \log (1/\delta)$. But $t = \sqrt{n}$ leading to a contraction. Therefore, we cannot have any guarantee when $\delta \approx 2^{-\sqrt{n}}$.

This paper leads to many open problems. We conclude the paper by listing few open questions. 
\begin{enumerate}
	\item The first open problem is to give a distribution that uses linear random samples, allows fast matrix-vector multiplication, and allows embedding for all possible values of $r$. As shown in the above counterexample, the main source of slackness is in proving~\lemref{P_2}. If we can somehow improve this bound by some other technique which does not require lot of random bits, it would lead to an improvement of the overall bound. One technique that is often useful in such scenario is decoupling where under limited dependence, we prove results equivalent to that of completely independence of random variables. The excellent book by De la Pena and Gine~\cite{decoupling} shows many such examples. It would be interesting if any such methods can be applied in our context.
	\item An open problem, also suggested by Blocki {\it et al.}~\cite{BBDS12}, is to get a better error bound by, maybe, using error correcting codes. To quote them, 
	\begin{quote}
		Can we introduce some error-correction scheme to the problem without increasing r significantly? Error amplification without increasing r will allow us to keep the additive error fairly small. One can view the laplacian of a graph as a coding of answers to all $2^n$ cut-queries which is guaranteed to have at least a constant fraction of the code correct, in the sense that we get an approximation to the true cut-query answer. 
	\end{quote}
	\item At present, our understanding of low dimension embedding for any normed space other than $\ell_2$-norm is very restricted. For example, even modest question like whether one can embed any $m$-point subset of $L_1$ into a subpace of $L_1$ with dimension $cn$ for $c<1/2$ with $O(1)$ distortion is unclear. There are some positive works in this area, most notably by Schechtman~\cite{Sch87}, but since then, there  has been no significant progress in this area. One area of future research is to better understand the low-dimension embedding in other normed spaces.
\end{enumerate}

\end{versionC}


\bibliographystyle{plain}
{
\bibliography{linear}}

\begin{appendix}

\section{Reducing the Running time to $O(n \log r)$} \label{app:run}
To reduce the running time, we use the main result of~\cite{AC09} of which~\thmref{AC09} is a corollary. They proved that if $\mathbf{W}_b$ is a randomized Hadamard matrix of order $b$, then for any $\bz \in \bS^{b-1}$, 
\begin{align} \p \sparen{ \| \mathbf{W}_b \bz \|_\infty \geq \rho } \leq 2 b \exp(-\rho^2b/2). \label{eq:AC09} \end{align}
We use the same idea as used by Ailon and Liberty~\cite{AL09}, which is to replace $\bW$ by a block matrix with $n/b$ blocks of Hadamard matrix $\bW_b$ for $b = r^{c}$ for an arbitrary constant $c>1$:
\[  
\bW' := \begin{pmatrix} \bW_b & \mathbf{0}^b & \cdots &  \mathbf{0}^b \\  \mathbf{0}^b & \bW_b & \cdots &  \mathbf{0}^b \\ \vdots & \ddots & \ddots & \vdots \\  \mathbf{0}^b &  \mathbf{0}^b & \cdots  & \bW_b \end{pmatrix}
\]
The proof is similar to that of Ailon and Liberty~\cite{AL09} with further strengthening by Dasgupta, Kumar, and Sarlos~\cite{DKS10}. Let $\bD$ be a diagonal matrix as before. Then the random matrix $\bW' \bD$ is a orthogonal matrix. We form $\mathbf{G}$ by normalizing it. We need to study the behaviour of $\mathbf{G} \bx$ for a unit vector $\bx \in \bS^{n-1}$. Divide $\bx$ into $t=n/b$ equal blocks $\bx^{(1)}, \cdots, \bx^{(t)}.$ First of all notice that~\thmref{AC09} holds trivially when $\|\by\|_2 \leq 1$. Let $\mathbf{G}^{(i)}$ be the $i$-th block of $\mathbf{G}$. Then we can write $\by^{(i)} = \mathbf{G}^{(i)} \bx^{(i)}$. 
Moreover, if for a block $i$, if $\| \bx^{(i)}\|_2 \leq \sqrt{(\log (n/\delta))/n}$, then $\| \by^{(i)} \|_\infty \leq \sqrt{(\log (n/\delta)) /n} $. Since $\bx \in \bS^{n-1}$, there can be at most $n/\log (n/\delta)$ different blocks which have $\| \bx^{(i)} \|_2 \geq \sqrt{(\log (n/\delta)) /n}$. Plugging this in~\eqnref{AC09}, using the fact that $r^{c-1}/\log(n/\delta) \leq r^{c-1}/\log(r/\delta) \leq 1$, and applying union bound gives us the required bound.



\section{Differentially Private Streaming Algorithms of Upadhyay~\cite{Upadhyay14}} \label{app:stream}
In this section, we present the streaming algorithms of Upadhyay~\cite{Upadhyay14} for the sake of completion.  
\subsection{Matrix Multiplication} \label{sec:prod}
Suppose we want to compute matrix multiplication when matrices $\bA$ and $\mathbf{B}$ are streamed online.  The algorithm of Upadhyay~\cite{Upadhyay14} to compute the matrix multiplication is described below. The main idea is to lift the spectra of the input matrices above the threshold of~\thmref{DP}. 

\begin{description} 
	\item [{\sc Initialization.}] For multiplication approximation parameter  $\varepsilon$ with confidence $\delta$, privacy parameters $\alpha, \beta$, set $ r=O(\log(1/\delta)/\varepsilon^2)$. Set $w = {\sqrt{16r \ln(\frac{2}{\beta})}}/\alpha \ln(\frac{16r}{\beta})$. Set the intial sketches of $\bA$ and $\mathbf{B}$ to be all zero matrices $\Y_{\bA_0}$ and $\Y_{\bB_0}$.
	\item [{\sc Data-structure update.}]  Set $m=\max \{m_1,m_2\}$. On input a column $a$ of an $n \times m_1$ matrix $\bA$ and column $b$ of an $n \times {m_2}$ matrix $\mathbf{B}$ at time epoch $t$, set the column vector $\widehat{\bA}_{:a} = \begin{pmatrix} {w \mathbf{e}_a} & \mathbf{0}^{n+m} &   \bA_{:a} \end{pmatrix} $ and $\widehat{\mathbf{B}}_{:b} = \begin{pmatrix} w \mathbf{e}_b &   \mathbf{0}^{n+m}  &  \mathbf{B}_{:b}  \end{pmatrix} $. 
	 Compute  $\bPhi \widehat{\bA}_{:a}$ and $\bPhi \widehat{\mathbf{B}}_{:b}$, where $\bPhi$ is sampled as per the distribution defined in~\figref{basic}. Update the sketches by replacing the columns $a$ of $\Y_{\bA_{t-1}}$ and $b$ of $\Y_{\mathbf{B}_{t-1}}$ by the respective returned sketches to get the sketch $\Y_{\bA_t}, \Y_{\mathbf{B}_t}$.
	 \item [{\sc Answering matrix product.}] On request to compute the product at time $t$, compute $\Y_{\bA_t}^{\mathsf T} \Y_{\mathbf{B}_t}$.
\end{description}

\subsection{ Linear Regression} \label{sec:linear}
Suppose we want to compute linear regression  when the matrix $\bA$ is streamed online. The algorithm of Upadhyay~\cite{Upadhyay14} to compute linear regression is presented below. The main idea is to lift the singular values of the input matrix above the threshold of~\thmref{DP}. 

\begin{description}
	\item [{\sc Initialization.}] For multiplication approximation parameter  $\varepsilon$ and additive approximation parameter $\delta$, privacy parameters $\alpha, \beta$, set $r=O(m \log(1/\delta)/\varepsilon)$, $w =  {\sqrt{16r  \ln(2/\beta)}} \alpha^{-1} \ln(16r /\beta)$, and $\Y_{\bA_0}$ to be all zero matrix. 
	\item [{\sc Data-structure update.}] On input a column $c$ of an $n \times m$ matrix  $\bA$ at time epoch $t$, set  the column vector $\widehat{\bA}_{:c} =   \begin{pmatrix}  w\mathbf{e}_c  &  \mathbf{0}^{n+d}  &  \bA_{:c}  \end{pmatrix}$. Compute $\bPhi \widehat{\bA}_{:c}$, where $\bPhi$ is sampled as per the distribution defined in~\figref{basic}. Update the sketch of $\bA$ by replacing the column $c$ of $\Y_{A_{t-1}}$ by the  returned sketch to get the sketch $\Y_{A_t}$.
	\item [{\sc Answering queries.}] On being queried with a vector $\mathbf{b}_i$, set the column vector $\widehat{\mathbf{b}}_i=\begin{pmatrix}   \mathbf{e}_i  &  { \mathbf{0}^{n+d}}  &  \mathbf{b}_i  \end{pmatrix}$. Compute $\bPhi \widehat{\mathbf{b}}_i$ to get the sketch $\Y_{\mathbf{b}_i}$, where $\bPhi$ is sampled as per the distribution defined in~\figref{basic}. 	Compute a vector $\bx_i$ satisfying $\min_{\bx} \| \Y_{\bA_t} \bx_i-\Y_{\mathbf{b}_i}\|$.
\end{description}

\begin{versionA}
\section{Other Applications of our Distribution} \label{app:applications}
\subsection{Functional Analysis}
Charikar and Sahai~\cite{CS02} showed an impossibility result for dimension reduction using linear mapping in $\ell_1$ norm. Johnson and Naor~\cite{JN10} gave an elegant result which generalizes Charikar and Sahai's result. Concisely, they showed that a normed space which satisfies the {\sf JL} lemma is very close to being Euclidean in the sense that all its $n$-dimensional subspaces are isomorphic to Hilbert space with distortion $2^{2^{O(\log^* n)}}$. 
More formally, they showed the following.
\begin{theorem} \label{thm:JN10}
Let $\cX$ be a Banach space such that for every $n$ and every $\bx_1,\cdots , \bx_m \in \cX$, there exists a 
linear subspace $\mathbf{F} \in \cX$ of dimension at most $O(\varepsilon^{-2} \log (1/\delta))$ and a linear mapping from $\cX$ to $\mathbf{F}$
 such that~\eqnref{JL} is satisfied. Then for every $\ell \in \N$ and every $\ell$-dimensional subspace $\cE \subseteq \cX$, we have the Banach-Mazur distance between $\cE$ and Hilbert space  at most $2^{2^{O(\log^* \ell)}}$.
\end{theorem}

Our main result directly improves this characterization in the sense that we only need linear random samples. However, there are another subtle applications of our result in Functional analysis. For this, we need~\thmref{KW11}. 
Now, using the characterization by Johnson and Naor, our result gives an alternate reasoning for the impossibility result of~\cite{AGR14}. 
Moreover, Pisier~\cite{Pisier} proved that if a Banach space satisfies the condition that Banach-Mazur distance is $o(\log n)$, then the Banach space is super-reflexive. Because of the uniform convexity properties of the norm over these spaces~\cite{Finet86}, super-reflexive Banach spaces are of centre importance in Functional analysis (see for example the excellent book by Conway~\cite{Conway}). \thmref{newgaussian} and~\thmref{JN10} gives one such space using only linear random samples!

\section{Impossibility of Differential Privacy Using the Known Johnson-Lindenstrauss Transform} \label{app:notDP}
Blocki {\it et al.}~\cite{BBDS12} asked the question whether known distribution of matrices with {\sf JLP} preserves differential privacy or not. In this section, we answer their question in negation by giving simple counter-examples. Our counter-example works by showing two neighbouring matrices with the properties that they differ by at most one column of unit norm. Our counter-examples are constructed with care so as to allow the attack. We proceed by reverse-chronological order.

\section{Simpler Analysis of Known Johnson-Lindenstrauss Transform} \label{app:analysis}

\end{versionA}

\section{Proof of~\thmref{hw} with Gaussian Random Vectors} \label{app:gaussian}
Let $\by \sim \cN(0,1)^n$. Set $S = |\by^T \bB \by - \tr (\bB)|$ for a symmetric matrix $\bA$. Then using the Markov's inequality, we have the following.
\begin{align}
	\p [S > \eta] = \p [\exp(tS) < \exp(t \eta)] \leq \frac{\E[\exp(tS)]}{\exp(t \eta)}.  \label{eq:ts}
\end{align}
We need to evaluate $\E[\exp(tS)]$. Since $\bB$ is a symmetric matrix, we can consider its eigen-value decomposition, i.e., a  matrix $\mathbf{C}=\mathbf{U} \bB \mathbf{U}^T$ and $ \mathbf{w} =\mathbf{U}\by$ for an orthonormal matrix $U$ such that $\mathbf{C}$ is a diagonal matrix. Here we use the spherical symmetry of the difference of Gaussian distribution. We can then write 
\begin{align}
	\sum_{i,j}\bB_{ij} (\by_i \by_j -  \E[\by_i\by_j]) & = \by^T \bB \by - \tr(\bB) =  \mathbf{w} ^T \mathbf{C} \mathbf{w} -  \tr(\mathbf{C}) = \sum \mathbf{c}_i( \mathbf{w} _i^2-1). \label{eq:d_i}
\end{align}

To compute the expectation of this variable, we need to compute $\E[\exp(\mathbf{d}_i^2 ( \mathbf{w} _i^2-1))]$; the final expectation comes form the linearity of expectation. For that, we first need to estimate the $\PDF$ of the difference of i.d.d. Laplacian random variables. The following lemma was proven by Hanson~\cite[Lemma 3]{Hanson67}.
\begin{lemma} \label{lem:hw}
 	Let $X$ be any sub-exponential random variable. Then for $0<t\leq \tau$ for $\tau$, there is a universal constant $c$ such that	$\E[\exp(t(X-1))] \leq \exp(ct^2).$
\end{lemma}

We now return to the proof of~\thmref{hw}. It is well known that if $X$ is a subgaussian random variable, then $X^2$ is sub-exponential random variables. Therefore, we can use~\lemref{hw}. Recall that we wish to bound~\eqnref{ts}. We would have been done had  we not assumed that $|t| \leq \tau$ in~\lemref{hw}. Therefore, we have to work a little more to achieve the final bound. Plugging in the estimate of~\lemref{hw}, we can bound~\eqnref{ts} as follows:
\[  \exp(-\eta t) \exp \paren{ ct^2 \sum \mathbf{c}_i^2 } \]
for $t  \| \mathbf{d} \|_\infty \leq \tau.$
Therefore, we have 
\[  \p [S > \eta]  \leq \exp \paren{-\eta t + c t^2 \|\mathbf{C}\|_F^2}. \]

Simple calculus shows that the right hand side  is a concave up function and achieves minimum when $t = \eta/(2c \|\bB\|_F^2)$ if that is a permissible value for $t$. Now if we set $t' := \min \set{\eta/(2c\|\bB\|_F^2, \tau/ \|\bB\| }$. Then we have 
\[ \exp \paren{-\eta t + c t^2 \|\bB\|_F^2} \leq \exp(-t'(\eta- c t' \|\mathbf{C}\|_F^2)) \leq \exp(t' \eta/2) .\]

Setting constants $c_1= \tau/2$ and $c_2=1/4c$, we get the claim 
\begin{align*}
\p [|\by^T \bB \by - \tr (\bB)| \geq \eta ]  &\leq \exp \paren{ - \min \set{ \frac{c_1 \eta}{\|\bB\|}, \frac{c_2 \eta^2}{\|\bB\|_F^2} } } . 
\end{align*}

\end{appendix}

\end{document}